\documentclass[11pt]{article}

\usepackage[round]{natbib}
\usepackage[paperwidth=8.5in, paperheight=11in, margin=1in]{geometry}

\usepackage{float}
\usepackage{dsfont, amssymb,amsmath,amscd,latexsym, amsthm, amsxtra,amsfonts,enumerate}	

\allowdisplaybreaks                       
\theoremstyle{plain}

\usepackage[mathscr]{eucal}                 
\usepackage{dsfont}	                        

\usepackage{graphicx}
\usepackage{subfigure}                       
\usepackage[dvipsnames]{xcolor}				 
\usepackage[plainpages=false, pdfpagelabels]{hyperref} 
\hypersetup{
	colorlinks   = true,
	citecolor    = blue,
	linkcolor    = red,
	urlcolor     = orange
}
\usepackage{mathtools}                      
\mathtoolsset{showonlyrefs=true}          

\usepackage[normalem]{ulem}

\newtheorem{theorem}{Theorem}

\newtheorem{corollary}[theorem]{Corollary}
\newtheorem{definition}[theorem]{Definition}

\newtheorem{remark}[theorem]{Remark}

\numberwithin{equation}{section}
\numberwithin{theorem}{section}

\DeclareMathOperator*{\argmin}{arg\,min}

\newcommand\sig{\sigma}

\newcommand\lam{\lambda}

\newcommand\Ac{\mathcal{A}}

\newcommand\Uc{\mathcal{U}}

\newcommand{\Eb}{\mathbb{E}}

\newcommand{\Pb}{\mathbb{P}}

\newcommand{\Rb}{\mathbb{R}}
\newcommand{\Var}{\mathbb{V}}

\newcommand{\dd}{\mathrm{d}}
\newcommand{\ee}{\mathrm{e}}

\newcommand{\pb}{\bar{p}}

\newcommand{\Vb}{\bar{V}}
\newcommand{\thb}{\bar{\theta}}

\newcommand{\tho}{\theta_1^*}
\newcommand{\tht}{\theta_2^*}

\begin{document}
	
	\title{A Two-layer Stochastic Game Approach to  Reinsurance Contracting and Competition}
	
	\author{
		Zongxia Liang
		\thanks{Department of Mathematical Sciences, Tsinghua University, Beijing 100084, China.   E-mail: liangzongxia@tsinghua.edu.cn}
		\and
		Yi Xia
		\thanks{Corresponding author. Department of Mathematical Sciences, Tsinghua
			University, Beijing 100084, China. E-mail: xia-y20@mails.tsinghua.edu.cn}
		\and
		Bin Zou
		\thanks{Department of Mathematics, University of Connecticut, USA.
			E-mail: bin.zou@uconn.edu}
	}
	
	\date{\today\\Forthcoming in \emph{Insurance: Mathematics and Economics}}	
	\maketitle	
	\begin{abstract}
		We propose a two-layer stochastic game model to study reinsurance contracting and competition in a market with one insurer and two competing reinsurers. 
		The insurer negotiates with both reinsurers simultaneously for proportional reinsurance contracts that are priced using the variance premium principle. The reinsurance contracting between the insurer and each reinsurer is modeled as a Stackelberg game. 
		The two reinsurers compete for business from the insurer and optimize the so-called relative performance, instead of their own surplus, and their competition is settled by a noncooperative Nash game. We obtain a sufficient and necessary condition, related to the competition degrees of the two reinsurers, for the existence of an equilibrium. We show that the equilibrium, if exists, is unique, and the equilibrium strategy of each player is constant, fully characterized in semiclosed form. 
	Furthermore, we obtain interesting sensitivity results for the equilibrium strategies through both analytical and numerical studies.
	\end{abstract}
	

	%
	%

\section{Introduction}\label{sec:intro}

Insurers play a critical role in maintaining the financial stability of households and businesses, but they also face significant claim risks that can threaten their solvency. To mitigate these risks, insurers can diversify their insurance portfolios through reinsurance. Of course, when reinsurers offer reinsurance coverage to insurers, their goal is not solely to help insurers maintain solvency but also to generate profit from such businesses. As the customers of reinsurance are relatively limited,\footnote{{According to the statistics from the Insurance Information Institute, the top three insurance companies for automobile insurance by premiums written in the US in 2023---State Farm (18.3\%), Progressive (15.2\%), and Berkshire Hathaway Inc. (12.3\%)---together account for nearly 50\% of the total \$316.79 billion premiums from about 215 million motorists (insureds) with auto insurance in the US.}} especially when compared to those of regular insurance policies, reinsurers naturally compete for those limited customers (insurers), and the competition, in turn, motivates reinsurers to take into account the competitors' strategies in their own decision making. To capture the contract negotiation and business competition in the reinsurance market, we propose a two-layer stochastic game model with one insurer and two reinsurers and aim to obtain an equilibrium for such a complex game. 

The topic of optimal reinsurance is well studied in the actuarial literature, and papers on dynamic reinsurance often differ in contract types, risk models, premium rules, optimization criteria (preferences), and additional controls (such as investment and dividend decisions). 
The two dominant reinsurance contracts are excess-of-loss reinsurance (see \cite{asmussen2000optimal}) and proportional reinsurance, also called quota-share reinsurance (see \cite{schmidli2001optimal}). Regarding the insurer's risk exposure, a standard choice is the classical Cram\'er-Lundberg model (see \cite{schmidli2002minimizing}), but its approximating diffusion model is equally popular (see \cite{asmussen2000optimal} and \cite{schmidli2001optimal}). When it comes to the premium principles, the expected value principle is arguably the most common option (see \cite{schmidli2001optimal}), and the variance principle is another popular choice (see \cite{chi2012optimal} and \cite{liang2016optimal}). Further generalizations include the mean-CVaR premium principle (see \cite{tan2020optimal}) and extended distortion premium principles (see \cite{jin2024optimal}). Early works often adopt maximizing expected utility or minimizing ruin probability as their optimization criterion, but mean-variance preferences are also frequently used (see \cite{LRZ2015}). Several recent papers take into account ambiguity in modeling and apply different ambiguity preferences in the study (see \cite{GVY2018} for the worst-case approach and \cite{ZL2021} for the $\alpha$-maxmin preferences). Risk constraints, such as VaR and CVaR (CTE), can be incorporated so that the obtained optimal contracts help insurers meet the regulatory requirements in practice (see \cite{tan2009var} and \cite{LMWS2016}). Last, we mention that some recent contributions to this topic include novel features in their models; for instance, \cite{GVY2018} allow the insurer to invest their surplus in a financial market to explore statistical arbitrages caused by mispricing of stocks, and \cite{PCW2021} consider an insurer who possesses {insider information} on the asset prices. We refer interested readers to \cite{cai2020optimal} for a review article and \cite{albrecher2017reinsurance} for a monograph on optimal reinsurance.

Traditionally, research on optimal reinsurance takes the viewpoint of the insurer and seeks an optimal reinsurance contract that optimizes the insurer's objective; see those reviewed above for evidence. However, both parties of a reinsurance contract come to an agreement through bargaining and negotiations.
As such, to better describe the negotiation process, we should propose models that can capture the strategic interplay between the insurer and reinsurer of a contract. 
One obvious solution is to apply game theory and seek an equilibrium contract that takes into account the interests of both parties at the same time. 
Indeed, there has been a burgeoning interest to model reinsurance contracting as a game in recent years. \cite{JRYH2019} introduce a two-person cooperative game and solve it to obtain the equilibrium reinsurance contract. On the other hand, \cite{StackelbergLvYang, chen2019stochastic} model the reinsurance contracting as a Stackelberg game, which is a type of noncooperative game.

In this work, we adopt a game approach to study reinsurance contracting problems. In particular, we follow \cite{StackelbergLvYang} and model reinsurance contracting as a dynamic Stackelberg game. In such a game, the reinsurer is the leader and chooses the premium principle, while the insurer is the follower and chooses its reinsurance coverage. The hierarchical structure reflects the fact that reinsurers often have more advantages in negotiation; mathematically, this feature implies that the reinsurer knows the insurer's optimal decision and uses such information in determining its own optimal strategy (premium). We note that the (dynamic) Stackelberg game framework has already received considerable attention since \cite{StackelbergLvYang, chen2019stochastic}, which are likely the first two papers proposing such a framework. For instance,  \cite{DV2022} generalize \cite{chen2019stochastic} to a mean-variance premium principle and a random planning horizon; \cite{StackelbergBinZou} consider a general L\'evy risk process and allow the insurer and reinsurer to have heterogeneous, ambiguous beliefs on the risk distribution.   
\vskip 4pt 
However, the Stackelberg game model for reinsurance discussed earlier (see \cite{StackelbergLvYang, chen2019stochastic}, \cite{DV2022}, and \cite{StackelbergBinZou}) only considers the simple case of one insurer and one reinsurer, and thus cannot capture the observed fact that insurers often seek reinsurance coverage from \emph{multiple} reinsurers at the same time.\footnote{Reinsurers also compete with each other to win business from large insurers, as argued in \cite{reinsuranceTJ} and  \cite{MMT2023}.}
A minimum model that is consistent with this fact should consist of one insurer and \emph{two} reinsurers; one of such models is proposed in \cite{cao2023areinsurance}. In that paper, the insurer faces \emph{two} Stackelberg reinsurance games, and its optimal reinsurance strategies depend on the premium rules by both reinsurers in the market (one {reinsurer} applies the expected-value principle but the other uses the variance principle). As a result, although the two reinsurers in \cite{cao2023areinsurance} are \emph{not} directly linked, each impacts the other's decision through the common insurer.  \cite{cao2023areinsurance} extend the $2$-reinsurer model in \cite{cao2023breinsurance} to an $n$-reinsurer model, in which all reinsurers apply the variance premium principle. Both  \cite{cao2023breinsurance, cao2023areinsurance} adopt the ambiguity-averse and risk-neutral preferences from \cite{StackelbergBinZou} and obtain unique equilibrium strategies for all players. There are also recent works that propose a similar multi-player setup in their game model. \cite{GYS2020} study an asymmetric information linear-quadratic stochastic Stackelberg differential game with one leader and two followers; see also \cite{kroell2023optimal} for a model with multiple insurers (followers). 
\cite{LSC2024} consider an insurance market consisting of multiple competing insurers with a mean-field type interaction via the relative performance of their terminal wealth.

With the motivations discussed above, we are now ready to present the game model in this paper, which is largely inspired by the one in \cite{cao2023areinsurance}. The reinsurance market is composed of one representative insurer and two competing reinsurers, who both apply the variance premium principle,\footnote{\cite{chi2012optimal} studies optimal reinsurance under variance-related premium principles and argues that they form a crucial family of premium principles in actuarial science. We comment that the variance principle is indeed frequently used in the study of optimal reinsurance problems; see \cite{ZY2012}, \cite{CQSW2016}, \cite{liang2016optimal}, \cite{chen2019stochastic}, and \cite{cao2023breinsurance, cao2023areinsurance} for a short list.} but possibly with different loading factors, to price their reinsurance contracts. 
The insurer seeks proportional reinsurance contracts from both reinsurers simultaneously and thus faces two parallel Stackelberg reinsurance games. Given the variance loading factors $\theta = (\theta_1, \theta_2)$, the insurer's goal is to find its optimal reinsurance ceded proportions $(\bar{p}_1^\theta, \bar{p}_2^\theta)$, in which $\bar{p}_i^\theta$ is the optimal proportion ceded to Reinsurer $i$, $i=1,2$, and, as imagined, depends on the two reinsurers' joint premium strategy $\theta$. The two reinsurers in \cite{cao2023areinsurance} compete in an \emph{indirect} way, but here they compete \emph{directly}. To be precise, Reinsurer $i$ in our model does \emph{not} optimize its own surplus, $X_i(T)$, at the terminal time $T$, but compares its wealth to $X_j(T)$, the competitor, Reinsurer $j$'s wealth  ($i, j = 1,2$ and $i \neq j$). By following \cite{bensoussan2014class}, we define the so-called \emph{relative performance} $X_i(T) - \lam_i X_j(T)$, in which $\lam_i \ge 0$, and assume that Reinsurer $i$ optimizes its relative performance as a way to model \emph{direct} competition.
Therefore, the two reinsurers in our model interact via \emph{two} channels: the relative performance in their optimization criterion and the feedback response from the insurer; we note that \cite{cao2023areinsurance} only consider the latter channel. 
To settle the competition game, we resort to the notion of the classical noncooperative Nash game, in which the two reinsurers make decisions simultaneously.
To summarize, we propose a novel two-layer stochastic game model to study reinsurance contracting and competition, with the former by two parallel Stackelberg games and the latter by a noncooperative Nash game. All three players are utility maximizers, and we further assume that their preferences are given by an exponential, also called CARA (Constant Absolute Risk Aversion), utility function. Under the proposed game, the two reinsurers aim to find their equilibrium premium strategy $\theta_i^*$, $i=1,2$, and the insurer seeks an equilibrium reinsurance strategy $p^*=(p_1^*, p_2^*)$.

The key findings and contributions of this paper are discussed as follows. 
First, we propose a novel two-layer game model with multiple (two) reinsurers, which incorporates desirable features from at least three types of models: dynamic Stackelberg game models (see \cite{StackelbergLvYang}), models with multiple reinsurers (see \cite{cao2023areinsurance}), and game models with relative performance (see \cite{bensoussan2014class}).
Second, we obtain a sufficient and necessary condition for the existence of a game equilibrium, given explicitly by $0 \le \lam_1 \lam_2 < 1$, in which $\lam_i$ is the competition degree parameter of Reinsurer $i$ in its relative performance, $i=1,2$. Such a condition is precise and sharp because we show that if $\lam_1 \lam_2 \ge 1$, the proposed game admits no equilibrium. 
In addition, when $0 \le \lam_1 \lam_2 < 1$ holds, the equilibrium is unique, and all equilibrium strategies are constant. We are able to fully characterize the equilibrium strategies for all players in semiclosed form, subject to finding a unique fixed point of a bivariate function (such a task is easy from the computational point of view). 
Third, for the reinsurers' equilibrium premium strategies, we obtain \emph{analytical} results on the impact of risk aversion $\delta_i$ and competition degree $\lam_i$. We show that the increase of one player's risk aversion will cause \emph{both} reinsurers to charge a higher loading under equilibrium, but competition drives them to lower the premium. 
We also conduct a numerical study to investigate how those model parameters affect the insurer's  reinsurance decisions. The key findings are that the insurer cedes more risk to reinsurers when its own risk aversion increases or when the reinsurers are less risk averse.
\vskip 4pt
The rest of the paper is organized as follows. We introduce the two-layer Stackelberg-Nash game  model in Section \ref{sec:model}. In Section \ref{sec:eq}, we solve the equilibrium strategies for  all players in the Stackelberg game. We then analyze the equilibrium strategies mathematically and give the economic explanation in Section \ref{sec:econ}. Section \ref{sec:Conclusion} concludes this paper. Several  technical proofs are placed in Appendix \ref{sec:app}.

\section{Model}
\label{sec:model}

We consider a reinsurance market consisting of one insurer, labeled as Insurer (player) 0, and two \emph{competing} reinsurers, labeled as Reinsurer (player) 1 and Reinsurer (player) 2, over a finite horizon $[0, T]$, with $T > 0$ denoting the terminal time.  
To account for the competition between the two reinsurers, we assume that the market is formed under a \emph{tree} structure as in \cite{cao2023breinsurance, cao2023areinsurance}; see Figure \ref{fig:tree} for graphic illustration. 
Under such a market formulation, the insurer negotiates reinsurance contracts with both reinsurers \emph{simultaneously}, and the two reinsurers \emph{compete} for business from the insurer.
Note that our model includes two degenerate cases in which the insurer only purchases reinsurance from one reinsurer.
\begin{figure}[h]
	\centering
	\includegraphics[width=0.75 \textwidth, trim = 1in 1.5in 2.5in 1.5in, clip=true]{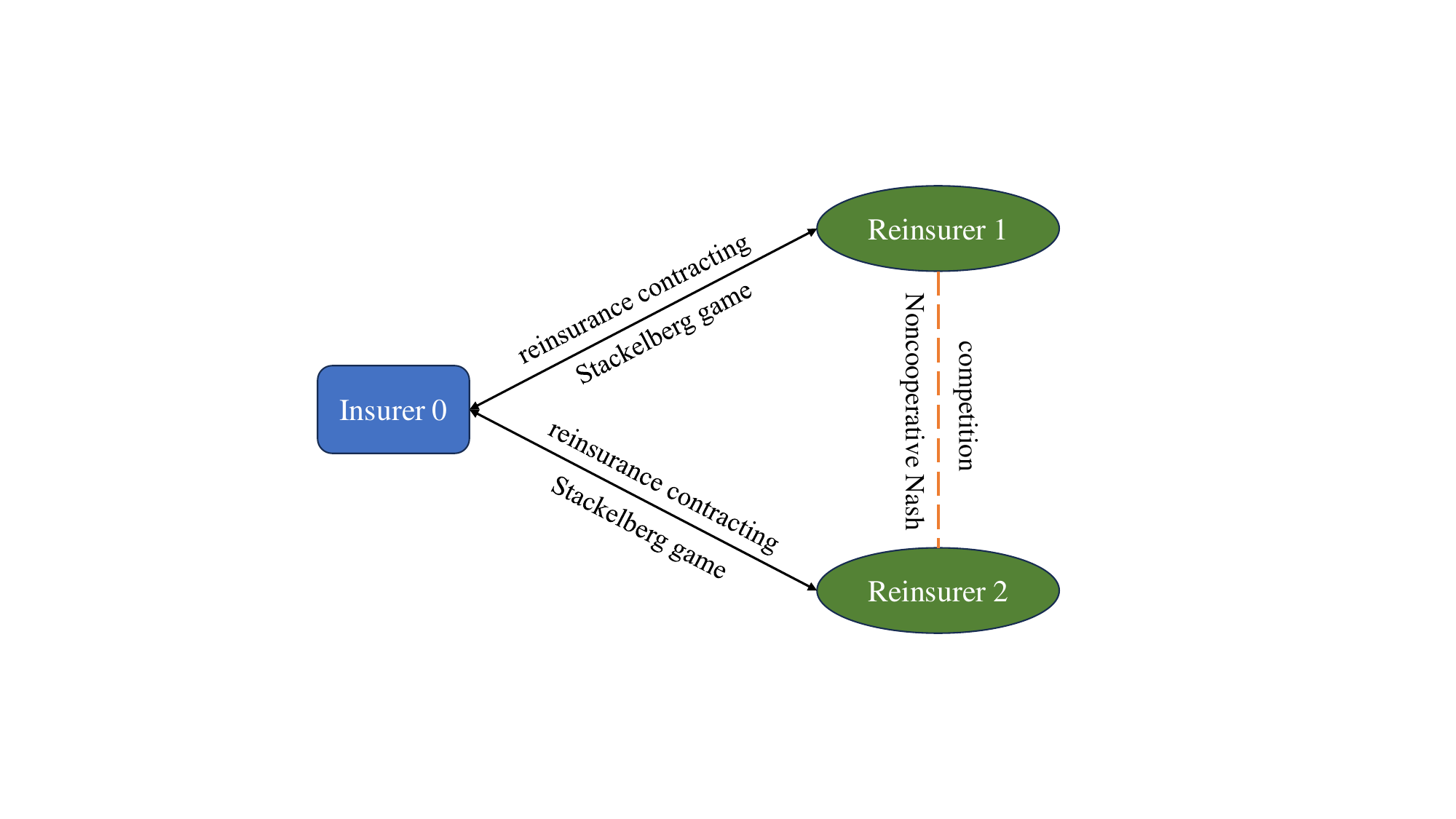}
	\caption{Two-layer Reinsurance Game Model}
	\label{fig:tree}
\end{figure}

\subsection{Strategies}

We assume that the insurer receives income at a constant rate $c > 0$ and is exposed to aggregate risks $L = \{L(t), t \in [0,T]\}$ that follow a standard diffusion model (see, e.g., \cite{schmidli2001optimal}).
To be precise, the dynamics of the risk exposure $L$ is governed by 
\begin{align*}
	\dd L(t) = \mu \, \dd t + \sigma \, \dd W(t) , 
\end{align*}
in which $\mu, \sigma > 0$ are the drift and volatility parameters of the insurer's risk process,  respectively, and  $W=\{W(t), t\in[0,T]\}$ is a one-dimensional Brownian motion defined on a filtered, complete probability space $\left(\Omega, \mathcal{F},\mathbb{F}, \mathbb{P}\right)$. Here, $\mathbb{F}= \left\{\mathcal{F}(t),{t \in[0, T]}\right\} $ is the augmentation of the natural filtration generated by $W$ and satisfies the usual conditions. Hereafter $\Eb(\cdot)$ and $\Var(\cdot)$ denote the expectation and variance operators under $\Pb$, respectively. 
We assume throughout the paper that $L(0) \gg 0$ and $\mu \gg \sigma$ (say $\mu > 3 \sigma$) under which $\Pb(L(t)>0)$ is approximately equal to 1 for all $t \in [0,T]$.

To mitigate the risk exposure, the insurer purchases \emph{proportional} reinsurance from Reinsurers 1 and 2, with ceded proportions $p_1 = \{p_1(t), t\in[0,T]\}$ and $p_2 = \{p_2(t), t\in[0,T]\}$, respectively.
That is, at each time $t \in [0,T]$, $p_1(t) \dd L(t)$ is ceded to Reinsurer 1, $p_2(t) \dd L(t)$ is ceded to Reinsurer 2, while $(1 - p_1(t) - p_2(t)) \dd L(t)$ is retained by the insurer.
To receive the reinsurance coverage, the insurer pays premiums that are computed by the variance premium principle (see, e.g., \cite{chi2012optimal} and \cite{liang2016optimal}). 
Denoting $\theta_i = \{\theta_i(t), t\in[0,T]\}$ the loading factor of Reinsurer $i$, the premium $c_i = \{c_i(t), t \in [0,T]\}$ for a reinsurance contract with ceded proportion $p_i$ satisfies 
\begin{align*}
	c_i(t) \, \dd t = \Eb \big( p_i(t) \, \dd L(t) \big) + \theta_i(t) \, \Var \big( p_i(t) \, \dd L(t) \big), \quad i = 1, 2,
\end{align*}
implying that 
\begin{align}
	\label{eq:c_i}
	c_i (t) = \mu \, p_i(t) + \theta_i(t) \sigma^2  p_i^2(t), \quad i = 1, 2.
\end{align}
Let $p := (p_1, p_2)$ denote the insurer's proportional reinsurance strategy and $\theta = (\theta_1, \theta_2)$ the pair of premium loadings, in which $\theta_i$ is Reinsurer $i$'s premium strategy, $i=1,2$. Throughout the paper, we suppress the dependence of processes (e.g., $c_i$ and $X_i$) on strategy $p$ and/or $\theta$. Given $p$ and $\theta$, the insurer's surplus process $X_0 = \{X_0(t), t \in [0,t]\}$ follows the dynamics 
\!\begin{align}
	\dd X_0(t) &= c \, \dd t - \big(1 - p_1(t) - p_2(t) \big) \, \dd L(t) - \big(c_1(t) + c_2(t) \big) \, \dd t  \\
	&=\! \left(c \!-\! \mu \!-\! \theta_1(t) \sigma^2 p_1^2(t) \! -\! \theta_2(t) \sigma^2 p_2^2(t)   \right) \dd t\!- \!\big(1 - p_1(t)\! - \!p_2(t) \big) \, \dd W(t)
		\label{eq:X_0}
\end{align}
with $X_0 = x_0 \in \Rb$.
Similarly, we obtain the dynamics of Reinsurer $i$'s surplus process $X_i = \{X_i(t), t \in [0,t]\}$ by 
\begin{align}
	\label{eq:X_i}
	\dd X_i(t) = c_i(t) \, \dd t - p_i(t) \, \dd L(t) = \theta_i(t) \sig^2 p_i^2(t) \, \dd t - \sig p_i(t) \, \dd W(t),
\end{align}
with $X_i(0) = x_i \in \Rb$, for $i=1,2$.

We proceed to define the insurer's and reinsurers' admissible strategies below.

\begin{definition}
	\label{def:ad}
	We call $p = (p_1, p_2)$ an admissible (reinsurance) strategy for the insurer if it satisfies the following conditions:
	\begin{itemize}
		\item [(i)] $ p_1$ and $p_2 $ are   {deterministic functions};
		\item [(ii)] $ p_i(t)\in[0,1] $, $i=1,2$, and $ p_1(t)+ p_2(t)\in[0,1] $ almost surely (a.s.), for all $t \in [0,T]$. 
	\end{itemize}
We call $\theta_i$ an admissible (premium) strategy for Reinsurer $i$, $i=1,2$, if {it is a deterministic, positive, and uniformly bounded function.}

Let $\Uc_0$, $\Uc_1$, and $\Uc_2$ denote the set of all admissible strategies for the insurer (Insurer $0$), Reinsurer $1$, and Reinsurer $2$, respectively,\footnote{With slight abuse of notation, we also use $\Uc_i$, $i=0,1,2$, to denote the admissible set at each time $t \in [0,T]$.} and $\Uc := \Uc_0 \times \Uc_1 \times \Uc_2$ denote the joint admissible set of all three players.
\end{definition}

As stated in Section \ref{sec:intro} (and will be described in detail shortly), we consider a game model with three players; as such, the surplus process of one player not only depends on its own strategy but also on the other two players' strategies. For this reason, we will only consider strategies that are \emph{jointly} admissible to all players, $(p, \theta_1, \theta_2) \in \Uc_0 \times \Uc_1 \times \Uc_2$, in the subsequent analysis. Note that given a triplet $u:=(p, \theta_1, \theta_2) \in \Uc := \Uc_0 \times \Uc_1 \times \Uc_2$, the insurer's surplus equation \eqref{eq:X_0} admits a unique solution $X_0 = X_0^{x_0, u}$ for each $x_0 \in \Rb$, and the same is true for Reinsurer $i$'s surplus equation \eqref{eq:X_i} under $X_i(0) = x_i \in \Rb$, $i=1,2$.

\begin{remark}
    As seen from the above setup, we assume \emph{a priori} that the insurer purchases \emph{proportional} reinsurance contracts from the reinsurers, which we explain as follows. First, the \href{https://www.iii.org/publications/insurance-handbook/regulatory-and-financial-environment/background-on-reinsurance}{insurance handbook} of the Insurance Information Institute (III) writes that ``insurers split (a layer of risk) among a number of reinsurance companies each assuming a portion''. This observation from the industry offers direct support to our choice of proportional reinsurance. Second, as nicely pointed out by a referee, it is known from the optimal (re)insurance literature that proportional contracts are optimal under the variance premium principle in various setups (see, e.g., \cite{chen2019stochastic} for a Stackelberg reinsurance game with one reinsurer and \cite{meng2017note} for optimal reinsurance with multiple reinsurers). As such, those results provide some theoretical foundation for the assumption of proportional reinsurance from onset. Last, because we use the diffusion model for the insurer's risk process, considering other forms of reinsurance contracts, say excess of loss reinsurance, is not straightforward. Indeed, existing works under the diffusion risk model often assume the use of proportional contracts; see, e.g., \cite{schmidli2001optimal}.  To further investigate the form of optimal contracts, we may consider the Cram\'er-Lundberg (or more general (Poisson) random measure models) and allow the contract indemnity $I:=I(t, Z)$ to be a function of time $t$ and loss $Z$; see, e.g., \cite{cao2023breinsurance, cao2023areinsurance}. We leave this direction, along with allowing the insurer to invest in a financial market, for future research.
\end{remark}

\subsection{Games}

To start, we assume that all three players, one insurer and two reinsurers, in the market are utility maximizers and  further that the preferences of (Re)Insurer $i$ are characterized by an exponential function 
\begin{align}
	\label{eq:U}
	U_i (x) = - \frac{1}{\delta_i} \, \ee^{-\delta_i x}, \quad i = 0, 1, 2, 
\end{align} 
in which $\delta_i > 0$ is the (constant) absolute risk aversion parameter.

Following \cite{cao2023breinsurance, cao2023areinsurance}, we model the reinsurance contracting between the insurer and each of the two reinsurers as a \emph{Stackelberg game} and the competition between two reinsurers as a \emph{noncooperative Nash game}. In all games, the information is perfect and symmetric to both parties.
In each Stackelberg game, the reinsurer is the leader and chooses its premium strategy $\theta_i$, and the insurer is the follower and chooses its reinsurance strategy $p_i$.
However, different from these two papers, the competition here is modeled via the so-called \emph{relative performance}; to be precise, Reinsurer $i$ in \cite{cao2023areinsurance} optimizes its own surplus (performance) $X_i$, $i=1,2$, but here it optimizes the relative performance $X_i - \lambda_i X_j$, with $j \neq i$, in which $\lambda_i > 0$ measures the competition degree of Reinsurer $i$ relative to its competitor Reinsurer $j$. The limit case of $\lambda_1 = \lambda_2 = 0$ reduces to the scenario investigated in \cite{cao2023breinsurance, cao2023areinsurance} and is considered in Section \ref{sub:lam_zero}.

As stated above, the insurer, as the follower in each Stackelberg contracting game, seeks an optimal reinsurance strategy in response to the reinsurers' chosen premiums. This is formally defined below.

\begin{definition}[Insurer's Problem]
	\label{def:insurer}
	Let the reinsurers' premium strategies $\theta = (\theta_1, \theta_2) \in \Uc_1 \times \Uc_2$ be arbitrary but fixed, and $(t, x) \in [0, T] \times \Rb$. The insurer seeks an optimal reinsurance strategy {$\pb^\theta = (\pb_1^\theta, \pb_2^\theta)$} that maximizes its dynamic objective $J_0^\theta$ defined by 
	\begin{align}
		\label{eq:J_0}
		J_0^\theta(t, x; p) = \Eb[U_0(X_0(T)) | X_0(t) = x],
	\end{align}
in which 
$U_0$ is given by \eqref{eq:U} with $\delta_0 > 0$. 
Denote the insurer's value function by 
\begin{align}
	\label{eq:Vb_0}
	\Vb_0^\theta(t,x) = J_0^\theta(t,x; {\pb^\theta}) = \sup_{p \in \Uc_0} \, J_0^\theta(t,x; p).
\end{align}
\end{definition}

In the Stackelberg contracting game between Reinsurer $i$ and the insurer, Reinsurer $i$ is the game leader and chooses its premium loading $\theta_i$, $i=1,2$. Given the hierarchical structure of Stackelberg game, the surplus dynamics of Reinsurer $i$ in \eqref{eq:X_i} holds under the insurer's optimal reinsurance strategy ${\pb^\theta}$ defined in \eqref{eq:Vb_0}; to account for such a fact, we introduce a more precise notation $X_i^{{\pb^\theta}}$ for Reinsurer $i$'s surplus process given $p = {\pb^\theta}$. As each reinsurer compares its own surplus with the competitor's, we define the relative performance $Y_i = \{Y_i(t), t \in [0, T]\}$ of Reinsurer $i$ by 
\begin{align}
	\label{eq:Y}
	Y_i(t) = X_i^{{\pb^\theta}}(t) - \lambda_i \, X_j^{{\pb^\theta}}(t), 
\end{align}
for $i, j = 1, 2$ and $i \neq j$, in which $\lambda_i > 0$ measures the competition degree of Reinsurer $i$ to its competitor Reinsurer $j$. Note that Reinsurer $i$'s relative performance, $Y_i$ in \eqref{eq:Y}, depends on not only its own premium strategy $\theta_i$ but also its competitor's $\theta_j$.
\begin{definition}
	[Reinsurers' Problem]
	\label{def:reinsurer}
	Let $i, j = 1,2$ with $i \neq j$,  $\theta_j \in \Uc_j$, and $(t, y) \in [0,T] \times \Rb$ be given. Reinsurer $i$ seeks an optimal premium strategy ${\thb_i^{\theta_j}}$ that maximizes its dynamic objective $J_i^{\theta_j}$ defined by 
	\begin{align}
		\label{eq:J_i}
		J_i^{\theta_j}(t, y; \theta_i) = \Eb \left[U_i(Y_i(T)) | Y_i(t) = y \right],
	\end{align}
in which $U_i$ is given by \eqref{eq:U} with $\delta_i > 0$, and $Y_i$ by \eqref{eq:Y}. Denote the value function by 
\begin{align}
	\label{eq:Vb_i}
	\Vb_i^{\theta_j}(t, y) = J_i^{\theta_j}(t, y; {\thb_i^{\theta_j}}) = \sup_{\theta_i \in \Uc_i} \, J_i^{\theta_j}(t, y; \theta_i)
\end{align}
for $i, j = 1,2$ and $i \neq j$.
\end{definition}

As is clear from Definition \ref{def:reinsurer}, Reinsurer $i$'s optimal strategy ${\thb_i^{\theta_j}}$ depends on its competitor's strategy $\theta_j$. The competition between them is settled by a noncooperative Nash game, which leads to the following definition of the two-layer game equilibrium.

\begin{definition}
	[Equilibrium]
	\label{def:eq}
	Assume that there exists a fixed point, denoted by $\theta^* :=(\theta_1^*, \theta_2^*)$, to the mapping $(\theta_1, \theta_2) \mapsto ({\thb_1^{\theta_2}}, {\thb_2^{\theta_1}})$, and it is admissible (i.e., $\theta^* \in \Uc_1 \times \Uc_2$). The equilibrium of the two-layer reinsurance contracting and competition game plotted in Figure \ref{fig:tree} consists of the following controls: 
	\begin{itemize}
		\item $\theta_1^* = {\thb_1^{\theta_2^*}}$ and $\theta_2^* = {\thb_2^{\theta_1^*}}$ are the equilibrium premium strategies of Reinsurer 1 and Reinsurer 2, respectively, and $\theta^* =(\theta_1^*, \theta_2^*)$ forms the Nash equilibrium of the competition game between the two reinsurers; 
		
		\item $p^* := (p_1^*, p_2^*) = ({\pb_1^{\theta^*}}, {\pb_2^{\theta^*}})$ is the insurer's equilibrium reinsurance strategy in the Stackelberg games with Reinsurers 1 and 2.
	\end{itemize}
\end{definition}

\section{Equilibrium}
\label{sec:eq}

In this section, we first solve the insurer's problem, for a given pair of admissible premiums $\theta = (\theta_1, \theta_2) \in \Uc_1 \times \Uc_2$, to obtain its optimal strategy $\pb^\theta$ in Section \ref{sub:insurer}; next, assuming that the insurer follows $\pb^\theta$, we solve each reinsurer's problem in Section \ref{sub:reinsurer}; finally, we combine the above results to obtain the equilibrium of the two-layer game model in Section \ref{sub:eq}.

\subsection{Insurer's Problem}
\label{sub:insurer} 

Recall that the insurer's problem is formally defined in Definition \ref{def:insurer}. 
To solve the insurer's problem in \eqref{eq:Vb_0}, we apply the dynamic programming method (see \cite{fleming2006controlled} and \cite{yong2012stochastic}) and present the solution in the theorem below.

\begin{theorem}
	\label{thm:insurer}
	For every $\theta=(\theta_1, \theta_2) \in \Uc_1 \times \Uc_2$, the insurer's optimal reinsurance strategy $\pb^\theta= \{(\pb_1^\theta (t), \pb_2^\theta (t)), t \in [0, T]\}$ is given by 
	\begin{align}
		\label{eq:pb_1}
		\pb_1^\theta(t) &=\dfrac{ \delta_0\theta_2(t)} {\delta_0 \theta_1(t) + \delta_0\theta_2(t) + 2 \theta_1(t) \theta_2(t) }, \\
		\pb_2^\theta(t) &=\dfrac{\delta_0 \theta_1(t)}{ \delta_0 \theta_1(t) + \delta_0\theta_2(t) + 2 \theta_1(t) \theta_2(t) }, \label{eq:pb_2}
	\end{align}
and its value function $\Vb_0^\theta$, defined in \eqref{eq:Vb_0}, is given by 
\begin{align}
	\label{eq:Vb0_sol}
	\Vb_0^\theta(t, x) = - \frac{1}{\delta_0} \, \ee^{-\delta_0 \, x + f_0(t)},
\end{align}
for every $(t,x) \in [0,T] \times \Rb$, 
in which $\delta_0 > 0$ is the insurer's absolute risk aversion, and $f_0$ is given by 
\begin{align}
 	f_0(t) = \delta_0\int_t^T \bigg\{   \mu-c+ \dfrac{ \delta_0\sigma^2\theta_1(s)\theta_2(s)} {\delta_0 \theta_1(s) + \delta_0\theta_2(s) + 2 \theta_1(s) \theta_2(s) } \bigg\}\dd s. \quad \label{eq:f0}
 \end{align}
\end{theorem}

\begin{proof}
	From the dynamics equation of $X_0$ in \eqref{eq:X_0}, we define the infinitesimal generator $\Ac_0^{(p, \theta)} \, v$, for every $v \in \mathrm{C}^{1,2}([0,T] \times \Rb)$ and every $(p, \theta) \in \Uc$, by 
	\begin{align}
	\Ac_0^{p, \theta} \, v(t,x) &= 	\partial_t v(t,x)+ \left(c-\mu- \theta_1(t)p_1^2(t)\sigma^2-\theta_2(t)p_2^2(t)\sigma^2 \right) \partial_x v(t,x)\\
	&\quad +\frac{1}{2} \left(1 - p_1(t) - p_2(t)\right)^2 \sigma^2 \partial_{xx} v(t,x),
	\end{align}
in which $\partial_{\cdot} v$ denotes the corresponding partial derivative of $v$ with respect to the subscript argument. Then, applying the dynamic programming principle, we obtain that the insurer's value function $\Vb_0^\theta$, if $\Vb_0^\theta \in \mathrm{C}^{1,2}([0,T] \times \Rb)$, is a classical solution to the following Hamilton-Jacobi-Bellman (HJB) equation
\begin{equation}
	\label{HJB-x0}
	\begin{cases}
		\sup\limits_{p \in \mathcal{U}_0} \, \mathcal{A}_0^{(p, \theta)} \, v(t,x)=0, \\
		v(T,x) = -\frac{1}{\delta_0}\exp(-\delta_0 x).
	\end{cases}
\end{equation}

To solve \eqref{HJB-x0}, we consider an ansatz in the form of \eqref{eq:Vb0_sol}, in which $f_0 \in \mathrm{C}^1([0,T])$ is yet to be determined. With the help of this particular ansatz, we reduce \eqref{HJB-x0} into 
\begin{align}
	\label{f0t}
	\begin{cases}
	\inf\limits_{p \in \Uc_0} \, \left\{\frac{1}{2} \sigma^2\delta_0^2 (1\!-\!p_1(t)\!-\!p_2(t))^2 \!-\! \delta_0 \left(c \!-\! \mu \!-\! \theta_1(t)p_1^2(t)\sigma^2\!-\!\theta_2(t)p_2^2(t)\sigma^2 \right)  \right\} 
	 \!+\! f_0'(t)\! = \!0, \\
		f_0(T) = 0.
	\end{cases}\quad 
\end{align}

A straightforward calculus shows that the minimization problem in \eqref{f0t} admits a unique minimizer $(\pb_1^\theta(t), \pb_2^\theta(t))$ given by 
\begin{align} 
	\pb_1^\theta(t)  =\frac{\delta_0}{\delta_0 + 2\theta_1(t)}(1- \pb_2^\theta(t)) \quad \text{and} \quad
	\pb_2^\theta(t)  =\frac{\delta_0}{\delta_0 + 2\theta_2(t)}(1- \pb_1^\theta(t)).
\end{align}
Solving the above system of equations immediately leads to the solutions in \eqref{eq:pb_1}-\eqref{eq:pb_2}. Note that for $\pb_1^\theta$ in \eqref{eq:pb_1} and $\pb_2^\theta$ in \eqref{eq:pb_2}, Condition $(ii)$ in the definition of admissible reinsurance strategies $p \in \Uc_0$ is satisfied; as a result, given $\theta \in \Uc_1 \times \Uc_2$, we have $\pb^\theta = (\pb_1^\theta, \pb_2^\theta) \in \Uc_0$. 

Plugging $\pb_1^\theta$ in \eqref{eq:pb_1} and $\pb_2^\theta$ in \eqref{eq:pb_2} into \eqref{f0t}, we easily solve the ordinary differential equation (ODE) of $f_0$ and obtain the solution as in \eqref{eq:f0}; 
note that $f_0 \in \mathrm{C}^1([0,T])$ as desired. Last, applying the standard verification arguments (see Theorem \ref{thm:veri_in}) confirms that all the results in Theorem \ref{thm:insurer} hold. 
\end{proof}

\begin{remark}
Theorem \ref{thm:insurer} shows that, even $\theta_1(t) < \theta_2(t)$ for all $t$ (i.e., reinsurance contracts offered by Reinsurer 1 are strictly cheaper than those offered by Reinsurer 2), the insurer may still purchase reinsurance from \emph{both} reinsurers at the same time. Such a result may seem puzzling at first look because intuition suggests that for the same ``product'', one should always buy from the cheaper supplier. We comment that this intuition is indeed correct if both reinsurers apply the expected value premium principle (see p.934 Section 2 in \cite{cao2023areinsurance}). However, this intuition does \emph{not} hold when both apply the variance premium principle, as is the case in this paper. This is because even with $\theta_1 < \theta_2$, a \emph{combination} of proportional contracts from both reinsurers may outperform buying only from Reinsurer 1; we refer interested readers to \cite{yao2024optimal} for a nice example.
\end{remark}

\subsection{Reinsurers' Problem}
\label{sub:reinsurer}

This subsection aims to solve each reinsurer's problem in the Stackelberg game, assuming that the premium strategy from the other reinsurer is fixed (see Definition \ref{def:reinsurer}).
Recall that in each of the two Stackelberg contracting games, the insurer is the follower, and the reinsurer is the leader. As such, the definition of Reinsurer $i$'s relative performance (surplus) $Y_i$ in \eqref{eq:Y} indicates that the insurer follows its optimal strategy $\pb^\theta$ obtained in Theorem \ref{thm:insurer}. Given $\theta = (\theta_1, \theta_2) \in \Uc_1 \times \Uc_2$, the dynamics of $Y_i$ is governed by 
\begin{align}
	\label{eq:dY}
	\dd Y_i(t)=\sigma^2 \left(\theta_i(t) \left(\pb_i^\theta(t)\right)^2-\lambda_j\theta_j(t) \left(\pb_j^\theta(t)\right)^2\right)\dd t-\sigma \big( \pb_i^\theta(t) - \lambda_j \pb_j^\theta(t) \big)\dd W(t) \; 
 \quad 
\end{align}
for $i, j = 1,2$ with $i \neq j$.
We present the key results below.

\begin{theorem}
	\label{thm:reinsurer}
	Let $i, j = 1, 2$ with $i \neq j$ and fix $\theta_j \in \Uc_j$. The optimal premium strategy 
 {$\thb_i^{\theta_j}$}
 for Reinsurer $i$ is given by 
	\begin{align}
		\label{eq:theta_op}
            {\thb_i^{\theta_j}(t) = \varphi_i(\theta_j(t)) ,}
	\end{align}  
in which the function $\varphi_i$, $i=1,2$, is defined by 
\begin{align}
	\label{eq:Theta}
	{\varphi_i(x)} = 	\dfrac{(\delta_0+2\delta_i) \, x^2 +(1+\lambda_j)\delta_0\delta_i \, x} {2 x^2 + ((1+2\lambda_j)\delta_0+2\lambda_j\delta_i) \, x +\lambda_j(1+\lambda_j)\delta_0\delta_i} \, ,
\end{align}
and the corresponding value function is given by 
\begin{align}
	\label{eq:Vbi_sol}
	\Vb_i^{\theta_j}(t, y) = -\frac{1}{\delta_i} \, \ee^{-\delta_i y + f_i(t)}, \quad (t, y) \in [0,T] \times \Rb,
\end{align}
in which $f_i$, $i=1,2$, is defined by 
\begin{align}
	f_i(t)\!=\!\int_t^T\! \bigg\{&\!-\delta_i\sigma^2 \left(\thb_i^{\theta_j}(s) \, \pb_i^2(s) \!-\!\lambda_j \thb_j^{\theta_i}(s) \,  \pb_j^2(s) \right)\! \\
 &+ \!\frac{1}{2} \sigma^2\delta_i^2 \big(\pb_i(s)\!-\!\lambda_j \pb_j(s) \big)^2\! \bigg\}\dd s \label{eq:f_i}
\end{align}
with $\pb_1:= \pb_1^{\thb_1^{\theta_2}, \theta_2}$ from \eqref{eq:pb_1} and $\pb_2 := \pb_2^{\theta_1, \thb_2^{\theta_1}}$ from \eqref{eq:pb_2}.
\end{theorem}

\begin{proof}
	This proof largely follows from that of Theorem \ref{thm:insurer}, and for that reason, we only provide a sketch below to save space. To start, based on \eqref{eq:dY}, we define an operator $\Ac_i^{(\theta_i, \theta_j)} v $, for every $v \in \mathrm{C}^{1,2}([0,T] \times \Rb)$ and every $(\pb^\theta, \theta) \in \Uc$, by 
	\begin{align}
		\Ac_i^{(\theta_i, \theta_j)} \, v(t, y) &= 	\partial_t v(t,y)+ \sigma^2 \left(\theta_i(t) \left(\pb_i^{\theta}(t)\right)^2 -\lambda_j \theta_j(t) \left(\pb_j^{\theta}(t)\right)^2 \right) \partial_y v(t, y)\\
		&\quad +\frac{1}{2} \, \sigma^2 \left(\pb_i^{\theta}(t)- \lambda_j \pb_j^{\theta}(t) \right)^2  \, \partial_{yy} v(t, y),
	\end{align}
for $i = 1, 2$. It can be shown that the value function $\Vb_i^{\theta_j}$, assuming $\Vb_i^{\theta_j} \in \mathrm{C}^{1,2}([0,T] \times \Rb)$, satisfies the following HJB equation
\begin{equation}\label{HJB-yi}
	\left\{\begin{aligned}
		&\sup_{\theta_i \in \Uc_i} \mathcal{A}_i^{(\theta_i, \theta_j)} v(t,y)=0,\\
		&v(T,y)=-\frac{1}{\delta_i} \, \ee^{ -\delta_i y}.
	\end{aligned}\right.
\end{equation}
By a length calculus (results are available upon request), we prove that the optimization problem in \eqref{HJB-yi} admits a unique maximizer $\thb_i^{\theta_j}$, and it is given by \eqref{eq:theta_op}, which is positive and bounded for every $\theta_j \in \Uc_j$. We then use $\thb_i^{\theta_j}$ in \eqref{eq:theta_op} to reduce \eqref{HJB-yi} into an ODE of $f_i$, to which there exists a unique solution given by \eqref{eq:f_i}. Finally, by a standard verification argument (see Theorem \ref{thm:veri_in} for a similar result), we confirm that $\thb_i^{\theta_j}$ in \eqref{eq:theta_op} is the optimal premium strategy for Reinsurer $i$, and $\Vb_i^{\theta_j}$ in \eqref{eq:Vbi_sol} is the value function to Reinsurer $i$'s problem formulated in Definition \ref{def:reinsurer}.  
\end{proof}

\subsection{Equilibrium}
\label{sub:eq}

In the last step, we combine the results from Theorems \ref{thm:insurer} and \ref{thm:reinsurer} to obtain the equilibrium of the two-layer game model plotted in Figure \ref{fig:tree}.  Our findings are summarized in the next theorem.

\begin{theorem}
	\label{thm:eq}
	If $0 < \lambda_1 \lambda_2 < 1$, the bivariate mapping $(\theta_1, \theta_2) \mapsto (\varphi_1(\theta_2), \varphi_2(\theta_1))$, with $\varphi_1$ and $\varphi_2$ defined in \eqref{eq:Theta}, admits a unique  constant fixed point in $(0, +\infty)\times(0, +\infty)$, denoted by $(\tho, \tht)$, and the unique equilibrium consists of the following controls: 
	\begin{itemize}
		\item Reinsurer $i$'s equilibrium premium strategy is \emph{constant} and equals $\theta_i^*$ (i.e., $\theta_i(t) \equiv \theta_i^*$ for all $t \in [0,T]$), for $i=1,2$.
		
		\item The insurer's equilibrium reinsurance strategy $p^* = (p_1^*, p_2^*)$ is \emph{constant} and equals $(\pb_1^{\theta^*}, \pb_2^{\theta^*})$ in \eqref{eq:pb_1}-\eqref{eq:pb_2} with $(\theta_1(t), \theta_2(t)) \equiv (\tho, \tht)$.
	\end{itemize}
	If $\lambda_1 \lambda_2 \ge 1$, there does not exist an equilibrium. 
\end{theorem}

\begin{proof}
	Based on the equilibrium definition (see Definition \ref{def:eq}) and Reinsurer $i$'s optimal strategy $\thb_i^{\theta_j} = \varphi_i(\theta_j)$ in \eqref{eq:theta_op}, the key is to find a fixed point for the bivariate mapping $(\theta_1, \theta_2) \mapsto (\varphi_1(\theta_2), \varphi_2(\theta_1))$, which, if exists, forms the pair of equilibrium premium strategies for the two reinsurers. 
	
	We start with analyzing the individual functions $\varphi_1$ and $\varphi_2$ defined by \eqref{eq:Theta}. We compute their derivatives as follows:
	\begin{align}
		\varphi'_i(x) &= \frac{(4\delta_0\delta_i\lambda_j+4\delta_i^2\lambda_j +2\lambda_j\delta_0^2+\delta_0^2)x^2 +2\delta_0\delta_i(\delta_0+2\delta_i)\lambda_j(1+\lambda_j)x}{\left(2x^2+((1+2\lambda_j)\delta_0+2\lambda_j\delta_i)x+\lambda_j(1+\lambda_j)\delta_0\delta_i\right)^2}  \\
		&\quad + \frac{\delta_0^2\delta_i^2\lambda_j(1+\lambda_j)^2}{\left(2x^2+((1+2\lambda_j)\delta_0+2\lambda_j\delta_i)x+\lambda_j(1+\lambda_j)\delta_0\delta_i\right)^2} > 0, \\
	\text{and} \quad 		\varphi''_i(x) &= -\frac{(16(\delta_0\!+\!\delta_i)\delta_i\lambda_j\!+\!4\delta_0^2(1\!+\!2\lambda_j))x^3\!+\!12\delta_0\delta_i(\delta_0\!+\!2\delta_i)\lambda_j(1\!+\!2\lambda_j)x^2\!}{\left(2x^2+((1+2\lambda_j)\delta_0+2\lambda_j\delta_i)x+\lambda_j(1+\lambda_j)\delta_0\delta_i\right)^3} \\
			&\quad - \frac{\!12\delta_0^2\delta_i^2\lambda_j(1\!+\!\lambda_j)^2x\!+\!2\delta_0^3\delta_i^2\lambda_j(1\!+\!\lambda_j)^3 }{\left(2x^2+((1+2\lambda_j)\delta_0+2\lambda_j\delta_i)x+\lambda_j(1+\lambda_j)\delta_0\delta_i\right)^3} <0 
	\end{align}
for $i = 1,2$ with $i \neq j$. We also have
\begin{align}
    \lim_{x \to 0} \, \varphi_i(x) = 0, \quad i = 1, 2.
\end{align}
Therefore, the inverse function of $\varphi_i$ exists over $(0,+\infty)$; let us denote it by $\varphi_i^{-1}$ hereafter. It is easy to see that, for $x>0$,
\begin{align}\label{eq:derivative2}
	\left(\varphi_i^{-1}\right)'(x) > 0 \quad \text{and} \quad  \left(\varphi_i^{-1}\right)''(x) > 0, \quad i = 1, 2.
\end{align}
Using the above derivative results on $\varphi_i$ and \eqref{eq:Theta}, we have the following limit results:
\begin{align}\label{eq:limit}
			&\lim_{x \to 0} \, \varphi_i'(x) = \frac{1}{\lambda_j}, 
	& &\lim_{x \to +\infty} \, \varphi_i'(x) = 0, \\&\lim_{x \to 0} \, \varphi_i^{-1}(x) = 0,   & & \lim_{x \to 0} \,\left(\varphi_i^{-1}\right)'(x) = \lambda_j,& &\lim_{x \to +\infty} \,\left(\varphi_i^{-1}\right)'(x) =  +\infty.  
\end{align}

Assume for now that $(\theta_1, \theta_2) \mapsto (\varphi_1(\theta_2), \varphi_2(\theta_1))$ has a fixed point, and we denote it by $(\tho, \tht)$. Then, by definition, $(\tho, \tht)$ is a solution to the following system of equations:
\begin{align}
	\label{eq:fixed_point}
	\varphi_2(\tho) = \varphi_1^{-1}(\tho) \quad \text{and} \quad \varphi_1(\tht) = \varphi_2^{-1}(\tht)
\end{align}
Given the properties of $\varphi_i$ and $\varphi_i^{-1}$, the system \eqref{eq:fixed_point} either has no solution or admits a unique, positive solution. Furthermore, the sufficient and necessary condition for the existence of a unique solution is 
\begin{align}
	\label{eq:deri_cond}
	\lim_{x \to 0} \, \varphi_1'(x) > \lim_{x \to 0} \, \left(\varphi_2^{-1}\right)'(x) \quad \text{and} \quad \lim_{x \to 0} \, \varphi_2'(x) > \lim_{x \to 0} \, \left(\varphi_1^{-1}\right)'(x), 
\end{align}
which is equivalent to  $\lambda_1 \lambda_2 < 1$ by recalling \eqref{eq:limit}.

The rest of the claims are evident by Theorems \ref{thm:insurer} and \ref{thm:reinsurer}.
\end{proof}

Recall that the quantity Reinsurer $i$ optimizes is \emph{not} its own terminal performance (surplus) $X_i(T)$ but rather the \emph{relative} performance $Y_i(T)$ defined in \eqref{eq:Y}. Although $\lambda_i \ge 1$ is allowed mathematically in the analysis, the term ``relative'' suggests that a more reasonable condition is that $0 < \lambda_i < 1$, under which a unique equilibrium exists and is characterized as in Theorem \ref{thm:eq}. As such, we assume $0 < \lambda_1 \lambda_2 < 1$ holds in the rest of the paper. 
Given such an assumption, Theorem \ref{thm:eq} provides a complete characterization of the equilibrium in semi-closed form, subject to solving a nonlinear system \eqref{eq:fixed_point} to obtain $\theta^*$. As there does not exist a closed-form solution for $\theta^*$, we focus on the sensitivity analysis of the equilibrium controls in the rest of this section. 

Recall from the definition of $\varphi_i$ in \eqref{eq:Theta} that both $\varphi_1$ and $\varphi_2$, and thus the equilibrium premium strategy $\theta^*$, depend on two sets of model parameters: (1) risk aversions ($\delta_0, \delta_1, \delta_2$) and (2) competition degrees ($\lambda_1$ and $\lambda_2$).
We present an analytical result regarding the impact of all these parameters on the equilibrium premium strategy $\theta^*$. 
The sensitivity results of $\theta_i^*$ in \eqref{eq:theta_sen} are consistent with our intuition and will be explained in detail in the next section, in which we also study how these model parameters affect the insurer's equilibrium reinsurance strategy $p^*$, for which an analytical result is not available.

\begin{corollary}
	\label{prop:sen}
	Suppose $0 < \lam_1 \lam_2 < 1$, and let $\theta^*_i$ be the equilibrium premium strategy of Reinsurer $i$, $i=1,2$, as obtained in Theorem \ref{thm:eq}. We have the following sensitivity results for the equilibrium premium strategy:
	\begin{align}
		\label{eq:theta_sen}
		\frac{\partial\theta_i^*}{\partial\delta_0}>0,\quad \frac{\partial\theta_i^*}{\partial\delta_i}>0,\quad \frac{\partial\theta_i^*}{\partial\delta_j}>0,\quad
		\frac{\partial\theta_i^*}{\partial\lambda_i}<0,\quad	\frac{\partial\theta_i^*}{\partial\lambda_j}<0
	\end{align}
	for $i , j = 1, 2$ with $i \neq j$. However, there does not exist any monotonicity result on the insurer's equilibrium reinsurance strategy $p^*$.
\end{corollary}

\begin{proof}
	See Appendix \ref{sec:app} for proof.
\end{proof}

{From Theorem \ref{thm:eq}, we know that equilibrium does \emph{not} exist when $\lam_1 \lam_2 = 1$, which includes the special case of $\lam_1 = \lam_2 = 1$, but a unique equilibrium exists when $\lam_1 \lam_2$ is \emph{strictly} less than 1. This drastically different ``boundary'' behavior motivates us to study the limit case of $\lam_1 \lam_2 \nearrow 1$; the results regarding each player's equilibrium strategy are obtained in Corollary \ref{cor:limit}. It is pleasing to see that, as $\lam_1 \lam_2$ increases to the boundary value 1, the equilibrium premium loading of both reinsurers reduces to 0. The economic meaning is that full competition ``forces'' both reinsurers to charge \emph{actuarially fair premium}, resulting in zero net (expected) profit, and the risk averse insurer benefits the most and buys full insurance, a classical result known from \cite{arrow1963uncertainty}. Corollary \ref{cor:limit} also helps explain why $\lam_1 \lam_2 = 1$ is the critical boundary for the existence of equilibrium shown in Theorem \ref{thm:eq}. Recall that both reinsurers are also risk averse; as $\lam_1 \lam_2 \nearrow 1$, their expected profit reduces to 0 (as $\theta_i^* \searrow 0$ for $i = 1, 2$), and they will simple walk away from the reinsurance business.
} 

\begin{corollary}
	\label{cor:limit}
Suppose $0 < \lam_1 \lam_2 < 1$. We have 
\begin{align}
	\lim_{\lam_1 \lam_2 \nearrow 1} \, \theta_1^* = \lim_{\lam_1 \lam_2 \nearrow 1} \, \theta_2^* = 0 \quad \text{and} \quad 
	\lim_{\lam_1 \lam_2 \nearrow 1} \, (p_1^* + p_2^*) = 1.
\end{align}
\end{corollary}

\begin{proof}
	See Appendix \ref{sec:app} for proof.
\end{proof}

We continue to investigate how competition affects each player's welfare in equilibrium, as measured by its value function. From the above two corollaries, we conjecture that the insurer's equilibrium value function is increasing with respect to $\lam_1$ and $\lam_2$, but the reinsurers' welfare decreases when competition intensifies. We analytically prove the former conjecture below and numerically confirm the latter in the next section.

\begin{corollary}
Suppose $0 < \lam_1 \lam_2 < 1$. The insurer's equilibrium value function is increasing with respect to both $\lam_1$ and $\lam_2$.
\end{corollary}

\begin{proof}
First,  using Theorems \ref{thm:insurer} and \ref{thm:eq}, we can show that the insurer's equilibrium value function is given by 
	\begin{align}
		\label{eq:V0star}
		V_0^*(t,x)=-\frac{1}{\delta_0}\ee^{-\delta_i x + f_0^*(t)}, \quad (t, x) \in [0, T] \times \Rb,
	\end{align}
	in which $f_0^*$ is defined by
	\begin{align}
		\label{f0star}
		f_0^*(t) = \delta_0 \bigg[   \mu-c+ \dfrac{ \delta_0\sigma^2\theta_1^*\theta_2^*} {\delta_0 \theta_1^* + \delta_0\theta_2^* + 2 \theta_1^* \theta_2^* } \bigg](T-t)
	\end{align}
with $\theta_1^*$ and $\theta_2^*$ being the unique equilibrium premium loadings of the two reinsurers established in Theorem \ref{thm:eq}.
	
From \eqref{f0star}, we directly compute 
\begin{align*}
	\frac{\partial f_0^*(t)}{\partial \theta^*_i}=\dfrac{ \delta_0^3\sigma^2(\theta_j^*)^2} {\left(\delta_0 \theta_1^* + \delta_0\theta_2^* + 2 \theta_1^* \theta_2^*\right)^2 }(T-t)>0,\quad t\in[0,T),
\end{align*}
which, combined with \eqref{eq:theta_sen}, implies 
\begin{align*}
	\frac{\partial f_0^*(t)}{\partial \lambda_i}<0, \quad i = 1, 2.
\end{align*}
The above result, along with the expression in \eqref{eq:V0star}, proves the desired claim.
\end{proof}

\subsection{The Special Case of $\lambda_1 \lambda_2 = 0$}
\label{sub:lam_zero}

The previous analysis assumes that $\lam_i>0$ for $i = 1, 2$ and shows that the game equilibrium exists if $0 < \lam_1 \lam_2 < 1$. In this subsection, we consider the special case of $\lam_1 \lam_2 = 0$; recall that the objective of the reinsurers in \cite{cao2023areinsurance} corresponds to the case of $\lam_1 = \lam_2 = 0$.

First, we consider the case with $\lam_i = 0$ but $\lam_j > 0$, in which $i, j = 1, 2$ and $i \neq j$. In this case, the function $\varphi_j$, originally defined in \eqref{eq:Theta} for $\lam_1, \lam_2 >0$, changes to 
\begin{align}
	\label{eq:Theta_j}
		\varphi_j(x) = \frac{(\delta_0+2\delta_j) x+\delta_0\delta_j}{2x+\delta_0}.
\end{align}
We easily obtain the following limit results of $\varphi_j$: 
\begin{align*}
	\lim_{x \to 0} \, \varphi_j(x) = \delta_j
	\quad \text{and} \quad 
	\lim_{x \to +\infty} \, \varphi_j(x) = \delta_j+\frac{\delta_0}{2},
\end{align*}
and derive the first- and second-order derivatives of $\varphi_j$ by 
\begin{align*}
	\varphi_j'(x) = \frac{\delta_0^2}{(2x+\delta_0)^2}>0
	\quad \text{and} \quad 
	\varphi_j''(x) = -\frac{4\delta_0^2}{(2x+\delta_0)^3}<0.
\end{align*}
As such, $\varphi_j$ has an inverse in the interval $ (\delta_j,\delta_j+\frac{\delta_0}{2})$, which we denote by $\varphi_j^{-1}$. Regarding this inverse, we have, for $x\in(\delta_j,\delta_j+\frac{\delta_0}{2})$,  
\begin{align}
	\left(\varphi_j^{-1}\right)'(x) > 0 \quad \text{and} \quad  \left(\varphi_j^{-1}\right)''(x) > 0,
\end{align}
and 
\begin{align}
	\lim_{x \to \delta_j} \, \varphi_j^{-1}(x)= 0 \quad \text{and} \quad  \lim_{x \to \delta_j+\frac{\delta_0}{2}} \, \varphi_j^{-1}(x)= +\infty.
\end{align}
Given the properties of $\varphi_i$ and $\varphi_j^{-1}$, the system \eqref{eq:fixed_point} admits a unique, positive solution $(\tho, \tht)$, in which $\theta_j^*\in (\delta_j,\delta_j+\frac{\delta_0}{2})$.

Second, we consider the case with $\lam_1 = \lam_2 = 0$. In this case, we obtain
\begin{align}
		\varphi_1(x) = \frac{(\delta_0+2\delta_1) x+\delta_0\delta_1}{2x+\delta_0}
		\quad \text{and} \quad
		\varphi_2(x) = \frac{(\delta_0+2\delta_2) x+\delta_0\delta_2}{2x+\delta_0}.
\end{align}
With the above $\varphi_1$ and $\varphi_2$, we can solve the system \eqref{eq:fixed_point} analytically and obtain the unique solution in closed form by
\begin{align}
	\label{eq:theta_eq_zero}
	\begin{cases}
			\theta_1^* = \dfrac{\delta_1}{2}+\dfrac{1}{2}\sqrt{\dfrac{\delta_0+\delta_1}{\delta_0+\delta_2}(\delta_0\delta_1+\delta_0\delta_2+\delta_1\delta_2)}, \\[2ex]
		\theta_2^* = \dfrac{\delta_2}{2}+\dfrac{1}{2}\sqrt{\dfrac{\delta_0+\delta_2}{\delta_0+\delta_1}(\delta_0\delta_1+\delta_0\delta_2+\delta_1\delta_2)}.
	\end{cases}
\end{align}

Therefore, the main result in Theorem \ref{thm:eq} obtained for $0 < \lam_1 \lam_2 < 1$ readily extends to the case when $\lam_1 \lam_2 = 0$, as summarized below. 

\begin{theorem}
	\label{thm:eq1}
	If $\lam_1 \lam_2 =0$, there exists a unique Nash equilibrium of premium strategies, $\theta^* = (\theta_1^*, \theta_2^*)$, for the two competing reinsurers, and the  insurer's equilibrium reinsurance strategy is given by $p^* = (\bar{p}_1^{\theta^*}, \bar{p}_2^{\theta^*})$. 
	In addition, the equilibrium premium pair $(\theta_1^*, \theta_2^*)$ is characterized by one of the following two cases:
	\begin{enumerate}
		\item If $\lam_i = 0$ but $\lam_j > 0$ with $i, j = 1, 2$ and $i \neq j$, then $(\theta_i^*, \theta_j^*)$ is the unique fixed point of the mapping $(\theta_i, \theta_j) \mapsto (\varphi_i(\theta_j), \varphi_j(\theta_i))$, in which $\varphi_i$ and $\varphi_j$ are defined by \eqref{eq:Theta} and \eqref{eq:Theta_j}, respectively. 
		
		\item If $\lam_1 = \lam_2 = 0$, then $(\theta_1^*, \theta_2^*)$ is given by \eqref{eq:theta_eq_zero}. 
	\end{enumerate}
\end{theorem}

\section{Economic Study}
\label{sec:econ}

We have solved the two-layer reinsurance contracting and competition game and obtained its equilibrium in Theorem \ref{thm:eq}. The characterization of all equilibrium controls is complete, subject to finding the fixed point, $\theta^*=(\theta_1^*, \theta_2^*)$, of the mapping $(\theta_1, \theta_2) \mapsto (\varphi_1(\theta_2), \varphi_2(\theta_1))$. However, finding such a $\theta^*$ can only be done by numerical methods in general. In this section, we conduct a thorough economic study focusing on the sensitivity analysis of the equilibrium controls, in particular, the insurer's equilibrium reinsurance strategy $p^*$.  

In the subsequent analysis, the default model parameters are set as follows:
\begin{table}[h]
	\centering
	\begin{tabular}{ccccc}\hline 
		$\delta_0$ & $\delta_1$ & $\delta_2$ & $\lam_1$ & $\lam_2$ \\ \hline 
		5 & 4 & 6 & 0.3 & 0.7 \\ \hline 
	\end{tabular}
\caption{Default Model Parameters}
\label{tab:para}
\end{table}
\\
Note from Table \ref{tab:para} that Reinsurer $2$ exhibits higher risk aversion and stronger competition degree than Reinsurer $1$ does. Also, $0 < \lam_1 \lam_2 < 1$ is satisfied, implying that there exists a unique equilibrium by Theorem \ref{thm:eq}. When we investigate the impact of one particular parameter, say $\delta_0$, we will allow it to vary over a reasonable range but keep the rest parameters fixed as in Table \ref{tab:para}.

\subsection{Sensitivity Analysis for the Insurer}

We first conduct a numerical study to investigate how model parameters affect the insurer's equilibrium strategy $p^* = (p_1^*, p_2^*)$, in which $p_i^*$ is the proportion ceded to Reinsurer $i$, $i=1,2$. Note that we were unable to derive any monotonicity results of $p^*$ in Proposition \ref{prop:sen}. 

In Figure \ref{pdelta0}, we plot how the insurer's equilibrium strategy $p_i^*$ changes with respect to its own risk aversion parameter $\delta_0$. It is pleasing to see that a more risk averse insurer cedes a larger proportion of its risk to reinsurers under equilibrium because doing so reduces the uncertainty of its terminal wealth $X_0(T)$ and yields a larger expected utility.\footnote{Such a result holds true under different Stackelberg game frameworks. For instance, when there is only one reinsurer, both \cite{StackelbergLvYang} and \cite{chen2019stochastic} arrive at a similar result; that is, the increase of the insurer's risk version leads to an increase of ceded risk to the reinsurer. Note that the reinsurer is a utility maximizer in \cite{StackelbergLvYang} but a mean-variance agent in \cite{chen2019stochastic}. This result can be further extended from risk aversion to ambiguity aversion, as shown in \cite{StackelbergBinZou}.} 
We observe from Figure \ref{pdelta0} that $p_1^* > p_2^*$; recall from Table \ref{tab:para} that $\delta_1 < \delta_2$. This is because the reinsurance policy offered by a more risk averse reinsurer is more expensive than the one offered by a less risk averse reinsurer, and the insurer purchases more cheaper contracts. 
	
\begin{figure}[H]
	\centering
	\includegraphics[width=0.48\linewidth, trim = 1cm 0cm 1cm 1cm, clip=true]{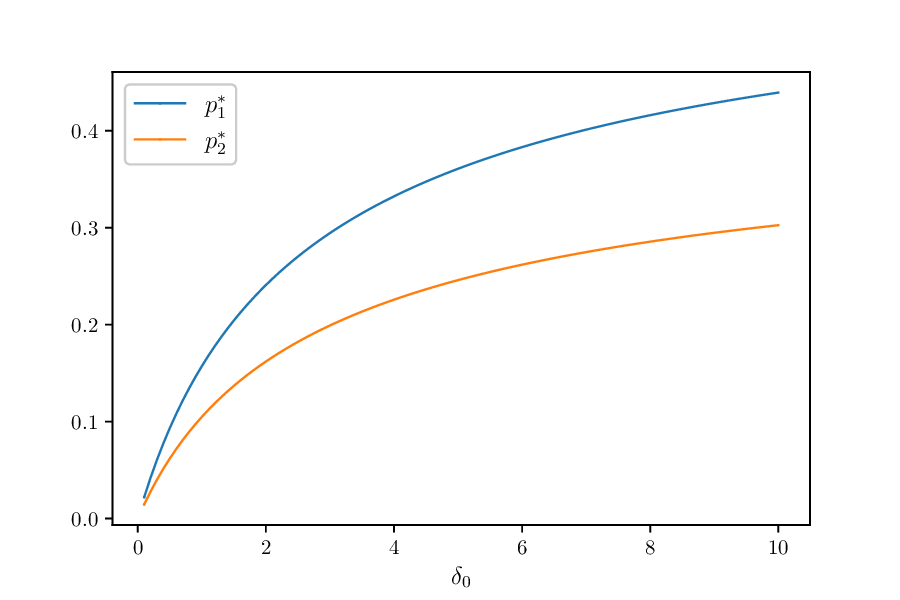}
 \\[-3ex]
	\caption{Impact of $ \delta_0 $ on $ p_1^* $ and $ p_2^* $}
	\label{pdelta0}
\end{figure}

We next study how the risk aversion of the two reinsurers, $\delta_1$ and $\delta_2$, affect the insurer's decision and plot $p_i^*$ as a function of these parameters in Figures \ref{pdelta1} and \ref{pdelta2}, respectively. We observe that when the risk aversion coefficient of one reinsurer
increases, the proportion of reinsurance ceded to this reinsurer will decrease significantly (see the blue curve in Figure \ref{pdelta1} and the orange curve in Figure \ref{pdelta2}); however, the proportion of reinsurance ceded to the other reinsurer may stay almost unchanged (Figure \ref{pdelta1})
or decrease slightly (Figure \ref{pdelta2}). The first result can be easily understood as follows. When $\delta_i$ increases, the contract offered by Reinsurer $i$ becomes more expensive, and the insurer reacts by ceding less risk to Reinsurer $i$, leading to a decrease of $p_i^*$. However, the impact of $\delta_i$ on $p_j^*$, $i \neq j$, is more subtle because although Reinsurer $j$ will raise its premium $\theta_j^*$, the increase may be at a much slower pace  than Reinsurer $i$.

\begin{figure}[H]
	\begin{minipage}[t]{0.48\textwidth}
		\centering
		\includegraphics[width=\linewidth]{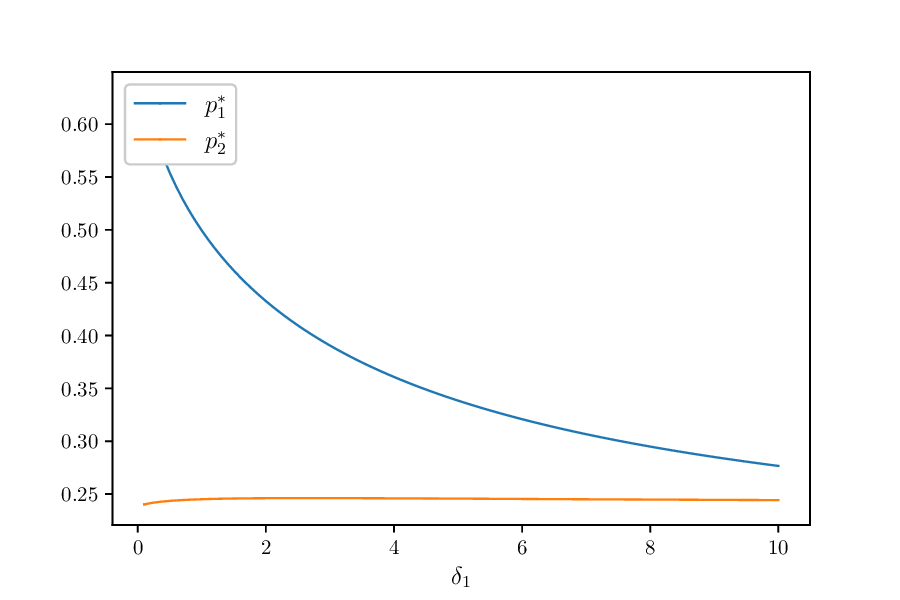}
  \\[-3ex]
		\caption{Impact of $ \delta_1 $ on $ p_1^* $ and $ p_2^* $}
		\label{pdelta1}
	\end{minipage}
	\begin{minipage}[t]{0.48\textwidth}
		\centering
		\includegraphics[width=\linewidth]{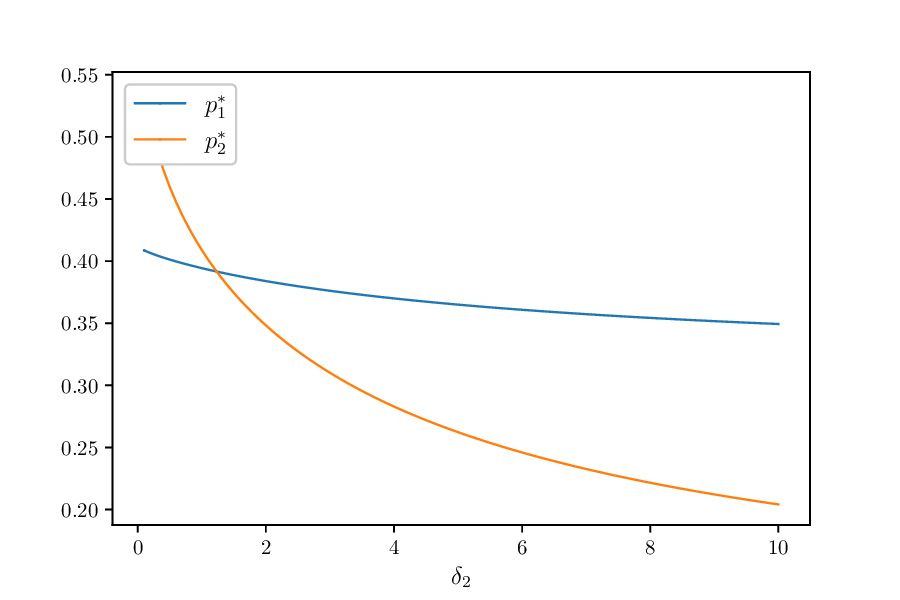}
  \\[-3ex]
		\caption{Impact of $ \delta_2 $ on $ p_1^* $ and $ p_2^* $}
		\label{pdelta2}
	\end{minipage}
\end{figure}

In the last step, we shift our attention to the impact of $\lam_1$ and $\lam_2$ on the insurer's equilibrium strategies $ p_1^* $ and $ p_2^* $. Recall that  $\lam_i$ measures the competition degree of Reinsurer $i$ to its competitor Reinsurer $j$ in the relative performance (see \eqref{eq:Y}). 
Both Figures \ref{plambda1} and \ref{plambda2} show that an increase in the competition degree of one reinsurer pushes the insurer to purchase more coverage from \emph{both} reinsurers, albeit at very different scales. This is because when $\lam_i$ increases, more intense competition drives both reinsurers to lower their premiums, which in turn incentivizes the insurer to cede more of its risk away to the reinsurers.

\begin{figure}[tbph]
	\begin{minipage}[t]{0.48\textwidth}
		\centering
		\includegraphics[width=\linewidth]{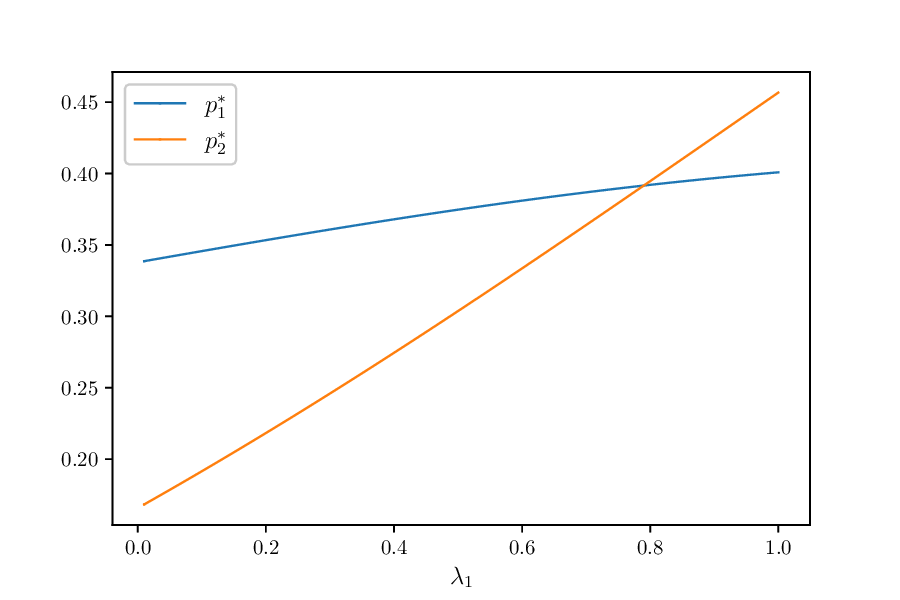}
  \\[-3ex]
		\caption{Impact of $ \lambda_1 $ on $ p_1^* $ and $ p_2^* $}
		\label{plambda1}
	\end{minipage}
	\begin{minipage}[t]{0.48\textwidth}
		\centering
		\includegraphics[width=\linewidth]{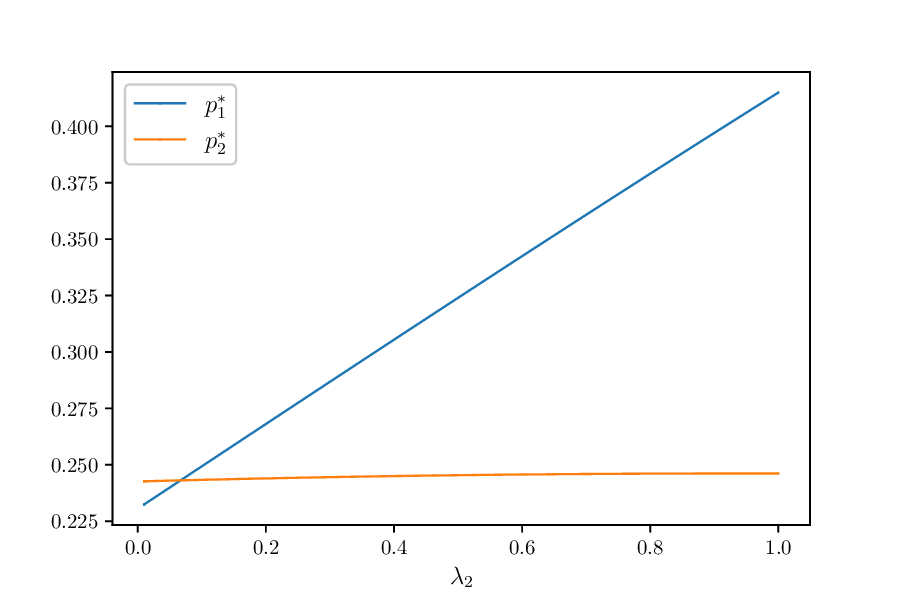}
  \\[-3ex]
		\caption{Impact of $ \lambda_2 $ on $ p_1^* $ and $ p_2^* $}
		\label{plambda2}
	\end{minipage}
\end{figure}

\subsection{Sensitivity Analysis for the Reinsurers}

We next turn our attention to the two reinsurers and investigate how model parameters affect their equilibrium strategies $\theta_1^*$ and $\theta_2^*$. Keep in mind that we have already obtained analytical results on the sensitivity of $\theta_1^*$ and $\theta_2^*$ in Proposition \ref{prop:sen}.

In Figures \ref{thetadelta0}-\ref{thetadelta2}, we plot how each reinsurer's equilibrium premium strategy $\theta_i^*$, $i=1,2$, reacts to the change of one player's risk aversion ($\delta_0$, $\delta_1$, or $\delta_2$). The plots in all figures confirm the positive relation between $\theta_i^*$ and $\delta_j$ established in \eqref{eq:theta_sen}. A larger $\delta_0$ implies a higher demand from the insurer for reinsurance coverage; upon recognizing this, both reinsurers explore such an opportunity by increasing premiums to achieve a higher expected utility. When the risk aversion $\delta_i$ of Reinsurer $i$ increases, one naturally anticipates an increase in its own premium $\theta_i^*$ because this helps reduce the risk taking by Reinsurer $i$ from the insurer. Interestingly, we observe that the increase of $\delta_i$ pushes $\theta_j^*$, the competitor's equilibrium premium, up; such a result is likely due to the direct competition implied by the relative performance used in optimization.
Last, we observe $\theta_2^* > \theta_1^*$ in most scenarios, except for very small $\delta_2$ in Figure \ref{thetadelta2}. This observation can be explained by the fact that Reinsurer 2 exhibits a higher degree of risk
aversion and places a greater emphasis on competition by choosing a larger $\lam_2$.

\begin{figure}[h]
	\centering
	\includegraphics[width=0.48\textwidth]{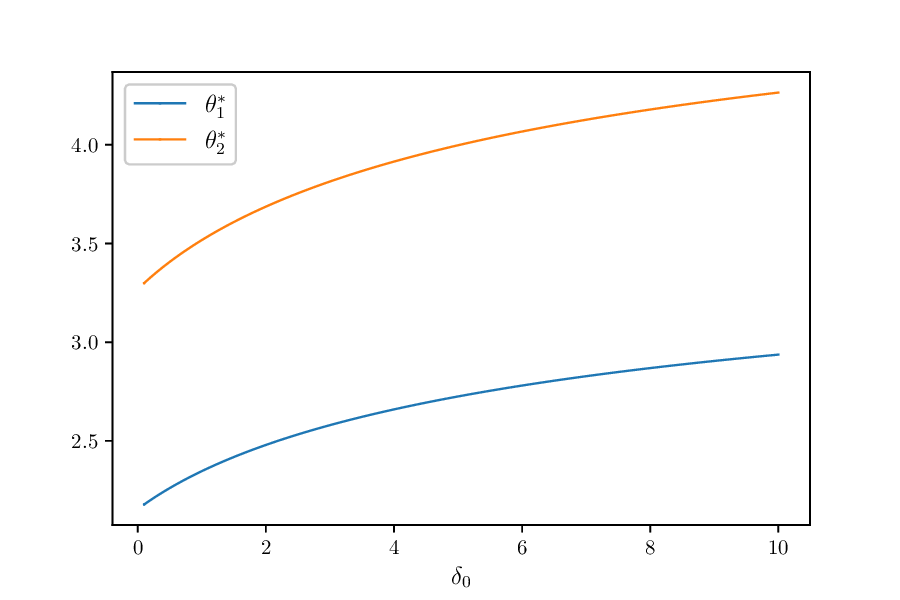}
 \\[-3ex]
	\caption{Impact of $ \delta_0 $ on $ \theta_1^* $ and $\theta_2^* $}
	\label{thetadelta0}
\end{figure} 

\begin{figure}[h]
	\begin{minipage}[t]{0.48\textwidth}
		\centering
		\includegraphics[width=\linewidth]{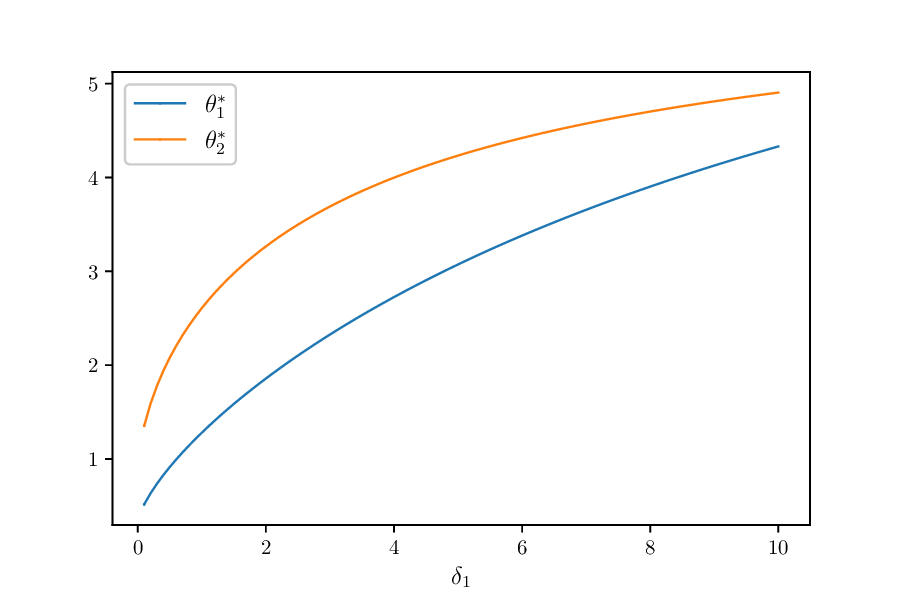}
  \\[-3ex]
		\caption{Impact of $ \delta_1 $ on $ \theta_1^* $ and $\theta_2^* $}
		\label{thetadelta1}
	\end{minipage}
	\begin{minipage}[t]{0.48\textwidth}
		\centering
		\includegraphics[width=\linewidth]{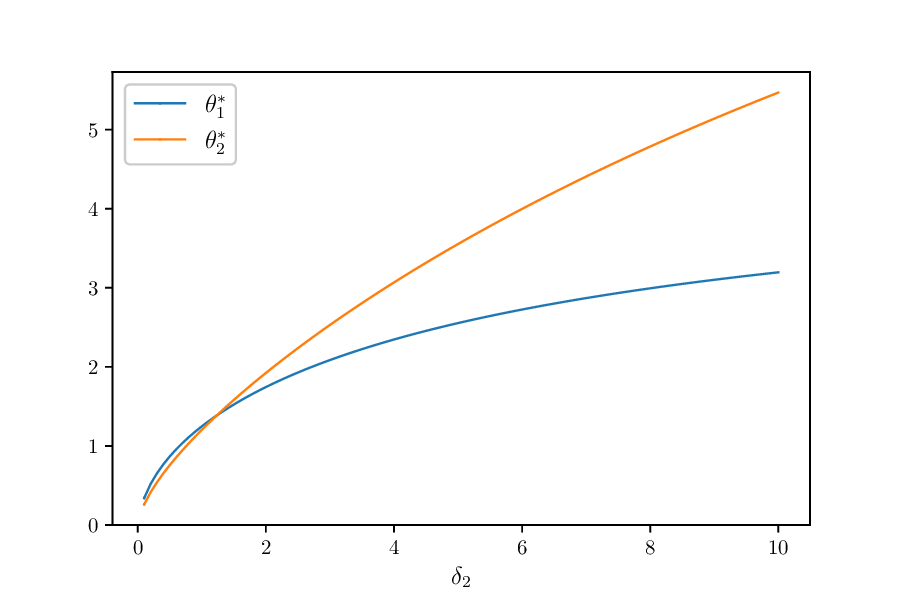}
  \\[-3ex]
		\caption{Impact of $ \delta_2 $ on $ \theta_1^* $ and $\theta_2^* $}
		\label{thetadelta2}
	\end{minipage}
\end{figure}

Next, we study the competition degree $\lam_i$ on the reinsurers' premium decision. 
Both Figures \ref{thetalambda1} and \ref{thetalambda2} show that when one reinsurer's competition degree  $\lam_i$ increases, \emph{both} reinsurers lower their premium loadings. This is due to the ``chain effect'' of the competition game explained as follows. A larger $\lam_i$ means that Reinsurer $i$ has a stronger incentive to compete with Reinsurer $j$, and to outperform its competitor Reinsurer $j$, Reinsurer $i$'s natural move is to lower its price by choosing a smaller $\theta_i^*$; as Reinsurer $j$ also takes into account Reinsurer $i$'s performance in evaluating its utility, it will reduce its loading $\theta_j^*$ in order to stay competitive to Reinsurer $i$. 

\begin{figure}[tbph]
	\begin{minipage}[t]{0.48\textwidth}
		\centering
		\includegraphics[width=\linewidth]{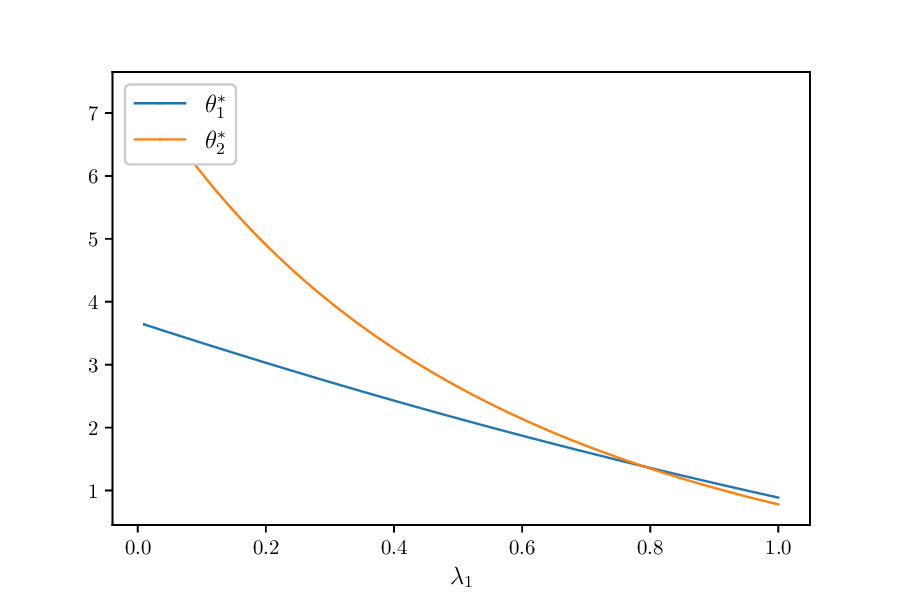}
  \\[-3ex]
		\caption{Impact of $ \lambda_1 $ on $ \theta_1^* $ and $\theta_2^* $}
		\label{thetalambda1}
	\end{minipage}
	\begin{minipage}[t]{0.48\textwidth}
		\centering
		\includegraphics[width=\linewidth]{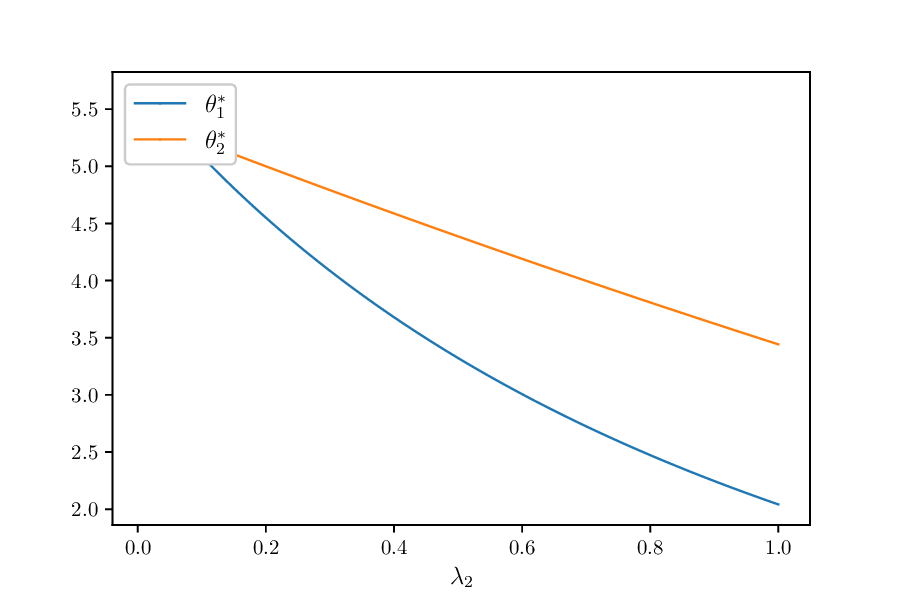}
  \\[-3ex]
		\caption{Impact of $ \lambda_2 $ on $ \theta_1^* $ and $\theta_2^* $}
		\label{thetalambda2}
	\end{minipage}
\end{figure}

Last, we study how competition, as measured by $\lam_i$, impacts the reinsurers' equilibrium value function. To that end, we first use Theorems \ref{thm:reinsurer} and \ref{thm:eq} to show that Reinsurer $i$'s equilibrium value function $V_i^*$, $i = 1, 2$, is given by 
	\begin{align}
		\label{eq:Vistar}
		V_i^*(t,x) = -\frac{1}{\delta_i}\ee^{-\delta_i x+f_i^*(t)}, \quad  (t, x) \in [0,T] \times \Rb, 
	\end{align}
in which $f_i^*$ is defined by
\begin{align}\label{eq:fi:star}
	f_i^*(t) = \dfrac{ \sigma^2\delta_0\delta_i(\theta_j^*-\lambda_j\theta_i^*)} {(\delta_0 \theta_1^* \!+\! \delta_0\theta_2^* \!+\! 2 \theta_1^* \theta_2^* )^2}\bigg[\!-\!\delta_0\theta_i^*\theta_j^*\!+\!\frac{1}{2}\delta_i(\theta_j^*\!-\!\lambda_j\theta_i^*)\bigg](T\!-\!t),
\end{align}
for $i,j=1,2$ and $i\neq j$.
From \eqref{eq:fi:star}, we obtain ($\propto$ means ``positively proportional to'')
\begin{align} f_i^*(t)\propto
\dfrac{ \theta_j^*-\lambda_j\theta_i^*} {(\delta_0 \theta_1^* \!+\! \delta_0\theta_2^* \!+\! 2 \theta_1^* \theta_2^* )^2}\bigg[\!-\!\delta_0\theta_i^*\theta_j^*\!+\!\frac{1}{2}\delta_i(\theta_j^*\!-\!\lambda_j\theta_i^*)\bigg].
 \end{align} 
With slight abuse of notation, define   
\begin{align} 
	\label{eq:fi_new}
	f_i^*=
\dfrac{ \theta_j^*-\lambda_j\theta_i^*} {(\delta_0 \theta_1^* \!+\! \delta_0\theta_2^* \!+\! 2 \theta_1^* \theta_2^* )^2}\bigg[\!-\!\delta_0\theta_i^*\theta_j^*\!+\!\frac{1}{2}\delta_i(\theta_j^*\!-\!\lambda_j\theta_i^*)\bigg],
 \end{align}
for $i,j=1,2$ and $i\neq j$. 
We plot the graphs of $f_1^*$ and $f_2^*$, defined by \eqref{eq:fi_new}, as a function of $\lam_1$ in Figure \ref{flambda1} and as a function of $\lam_2$ in Figure \ref{flambda2}. 
Recalling the inverse relation between $V^*_i$ and $f_i^*$ from \eqref{eq:Vistar}, we observe that as $\lam_1$ (or $\lam_2$) increases, \emph{both} $f_1^*$ and $ f_2^*$ increase, and in turn, both $V^*_1$ and $V^*_2$ decrease. As such, competition lowers both reinsurers' equilibrium value function.

 \begin{figure}[H]
	\begin{minipage}[t]{0.48\textwidth}
		\centering
		\includegraphics[width=\linewidth]{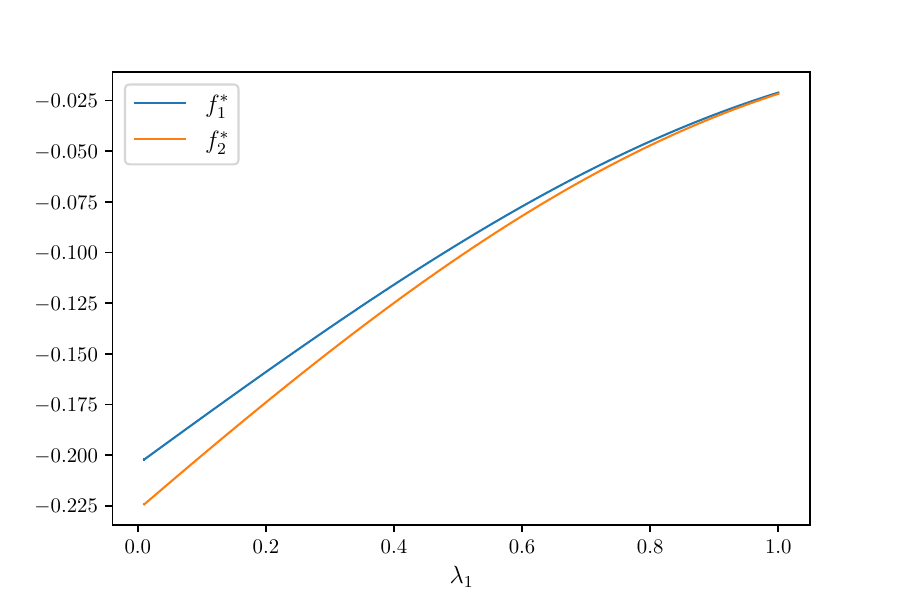}
  \\[-3ex]
		\caption{Impact of $ \lambda_1 $ on $ f_1^* $ and $ f_2^* $}
		\label{flambda1}
	\end{minipage}
	\begin{minipage}[t]{0.48\textwidth}
		\centering
		\includegraphics[width=\linewidth]{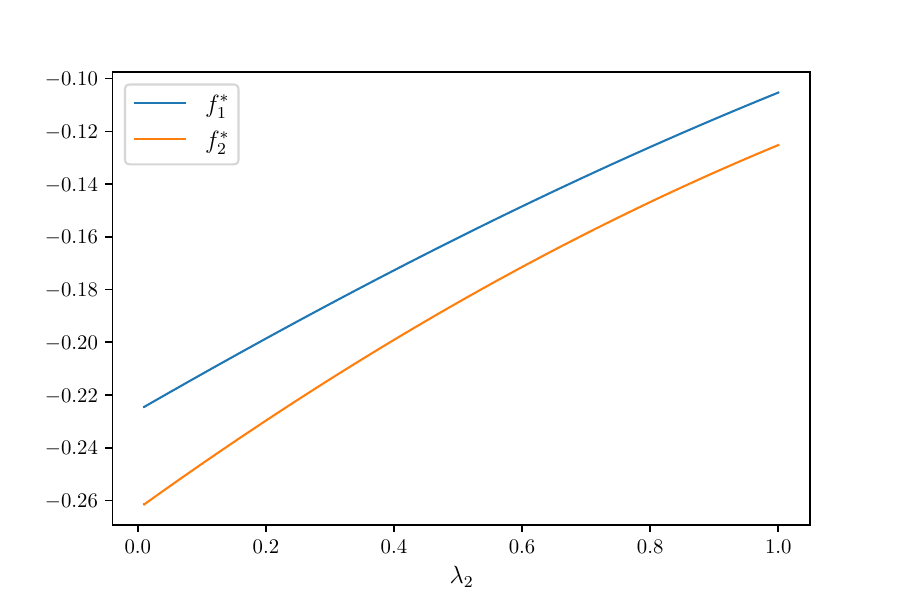}
  \\[-3ex]
		\caption{Impact of $ \lambda_2 $ on $ f_1^* $ and $ f_2^* $}
		\label{flambda2}
	\end{minipage}
\end{figure}

\section{Conclusion}
\label{sec:Conclusion}

In this paper, we introduce a two-layer stochastic game model to study reinsurance contracting and competition in a market with one insurer and two competing reinsurers. The reinsurance contracting problem between the insurer and each of the two reinsurers is captured by a Stackelberg game, in which the reinsurer is the leader and chooses the loading in its variance premium rule, and the insurer is the follower and chooses its proportional reinsurance strategy. 
Meanwhile, the competition between the two reinsurers is captured by the fact that each maximizes its expected utility of the relative performance at the terminal time. We solve such a complex game rigorously and completely. We show that a unique equilibrium exists and is fully characterized if $0 \le \lam_1 \lam_2 < 1$, and no equilibrium exists if $\lam_1 \lam_2 \ge 1$, in which $\lam_i \ge 0$ is the competition degree of Reinsurer $i$, $i = 1,2$. We also obtain several interesting qualitative results of the equilibrium strategies and conduct a numerical study to further study how insurer's and reinsurers' risk aversion coefficients, as well as competition parameters,  influence the equilibrium strategies.

\vspace{2ex}
\noindent
{\bf Acknowledgement.}  We would like  to thank two anonymous referees for valuable comments and the members of the group of Actuarial Sciences and Mathematical Finance  at the Department of Mathematical Sciences, Tsinghua University for their feedback. Zongxia Liang acknowledges the financial support from the National Natural Science Foundation of China (Nos.12271290, 12371477, 11871036).

\appendix
\section{Technical Results}
\label{sec:app}

We first provide a formal verification for the solution to the insurer's problem obtained in Theorem \ref{thm:insurer}. For a similar result, please refer to \cite{fleming2006controlled}[Theorem III.8.1, p.135].

\begin{theorem}
	[Verification theorem for the insurer's problem]
	\label{thm:veri_in}
	Let $v_0$ be a classical solution to the HJB equation \eqref{HJB-x0} associated with the insurer's problem in Definition \ref{def:insurer}. Then, for all $(t, x) \in [0,T] \times \Rb$ and all $\theta \in \Uc_1 \times \Uc_2$, we have
	\begin{itemize}
		\item[(a)] $v_0(x) \ge J_0^\theta(t,x; p)$ for every admissible strategy $p \in \Uc_0$.
		\item[(b)] If there exists an admissible $\pb^\theta \in \Uc_0$ such that 
		\begin{align*}
			\pb^\theta(s) = \argmin_{p(s) \in \Uc_0} \, \Ac_0^{(p, \theta)} \, v_0(s, \bar{X}^\theta_0(s))
		\end{align*}
	for almost all $s \in [t, T]$, in which $\bar{X}^\theta_0$ is the insurer's surplus under $\pb^\theta$, then 
	\begin{align}
		v_0(t, x) = \Vb_0^\theta(t,x) = J_0^\theta(t,x; \pb^\theta).
	\end{align}
	\end{itemize} 
\end{theorem}

\begin{proof}
	(a) Let $v_0$ be a classical solution to the HJB equation \eqref{HJB-x0}. For every $(p, \theta) \in \Uc$,  applying It\^{o}'s formula to $v_0(\cdot, X_0(\cdot))$, we get 
\begin{align*}
	v_0(T, X_0(T)) &= v_0(t, X_0(t)) + \int_t^T \mathcal{A}^{(p, \theta)}_0 v_0(s,X_0(s)) \, \dd s \\
	&\quad + \int_t^T \sigma  (1-p_1(s)-p_2(s)) \partial_x v_0(s,X_0(s)) \, \dd W(s).
\end{align*}
Introduce short-handed notations 	$\hat{\mu}(t) :=c-\mu- \sigma^2 \theta_1(t)p_1^2(t)- \sigma^2\theta_2(t)p_2^2(t)$ and $\hat{\sigma}(t) := \sigma(1-p_1(t)-p_2(t))$. Recall from Definition \ref{def:ad} that both $p$ and $\theta$ are uniformly bounded; as such, $f_0$ in \eqref{eq:f0}, $\hat{\mu}$, and $\hat{\sigma}$ are also uniformly bounded. That is, there exists a positive constant, denoted by $M$, such that 
\begin{align}
	\label{eq:bounds}
\sup_{t\in[0,T]} \ee^{2f_0(t)} \le M,\quad
\sup_{t\in[0,T]}\left|\hat{\mu}(t)\right| \le M,\quad 
\text{and} \quad \sup_{t\in[0,T]}\left|\hat{\sigma}(t)\right| \le M, \quad \text{a.s..}
\end{align}
Then, by a standard localization technique and using \eqref{HJB-x0}, we obtain 
\begin{align}
	\label{eq:v0_ineq}
	\Eb_{t,x}\left[v_0(T, X_0(T)) \right] = v_0(t, x) + \int_t^T \Eb_{t,x}\left[\mathcal{A}^{(p, \theta)}_0 v_0(s,X_0(s)) \right] \, \dd s \le v_0(t,x)
\end{align}
and, by the terminal condition in \eqref{HJB-x0} and the definition of $J_0^\theta$ in \eqref{eq:J_0}, 
\begin{align*}
	\Eb_{t,x}\left[v_0(T, X_0(T)) \right] = \Eb_{t,x} \left[ - \frac{1}{\delta_0} \, \ee^{-\delta_0 X_0(T)}\right] = J_0^\theta(t,x; p),
\end{align*}
in which $\Eb_{t,x}:= \Eb[\cdot | X_0(t) = x]$. 
Combining the above two results proves Assertion $(a)$.

Assertion $(b)$ is self-evident because the inequality in \eqref{eq:v0_ineq} becomes equality when $p = \pb^\theta$.
\end{proof}

We now apply Theorem \ref{thm:veri_in} to verify that the results in Theorem \ref{thm:insurer} hold as claimed. 
\begin{proof}[Verification for Theorem \ref{thm:insurer}]
	Consider a function $v_0$ given by 
	\begin{align*}
		v_0(t,x) = - \frac{1}{\delta_0} \, \ee^{-\delta_0 x}, \quad (t, x) \in [0,T] \times \Rb,
	\end{align*}
in which $f_0$ is defined in \eqref{eq:f0}. The proof of Theorem \ref{thm:insurer} shows that $v_0$ is a classical solution in the space of $\mathrm{C}^{1,2}([0,T] \times \Rb)$ to the HJB equation \eqref{HJB-x0}, then Assertion $(a)$ in Theorem \ref{thm:veri_in} holds. In fact, for such a $v_0$, the localization technique is not needed because we can directly show that the corresponding It\^o process is a martingale. To see that,  using \eqref{eq:bounds} and setting $X_0(0) = x_0$, we obtain
\begin{align*}
	\sup_{t\in[0,T]}\mathbb{E}\left[\ee^{ -2\delta_0X_0(t)} \right]
	=&\sup_{t\in[0,T]}\mathbb{E}\left[\exp\left( -2\delta_0x_0-2\delta_0\int_0^t\hat{\mu}(s)\dd s +2\delta_0\int_0^t\hat{\sigma}(s)\dd W(s)\right)\right]\\
	\leqslant& \exp\left( -2\delta_0x_0+2\delta_0 M  T\right)\sup_{t\in[0,T]}\mathbb{E}\left[\exp\left(2\delta_0\int_0^t\hat{\sigma}(s)\dd W(s)\right)\right]\\
	\leqslant& \exp\left( -2\delta_0x_0+2\delta_0 M T+2\delta_0^2 M^2 T\right),
\end{align*}
which in turn implies that
\begin{align*}
	&\mathbb{E}\left[\int_0^T\left|\hat{\sigma}(t) \partial_x v_0(t,X_0(t))\right|^2\dd t\right]
	\leqslant M^2 \sup_{t\in[0,T]}\mathbb{E}\left[\ee^{ -2\delta_0X_0(t)} \right] <\infty .
\end{align*}

From \eqref{eq:pb_1} and \eqref{eq:pb_2}, we easily see that 
\begin{align*}
	\pb_1^\theta(t) \in [0,1], \quad \pb^\theta_2(t) \in [0,1], \quad \text{and} \quad 1 - \pb^\theta_1(t) - \pb^\theta_2(t) \in [0,1],
\end{align*}
and $\pb^\theta = (\pb^\theta_1, \pb^\theta_2) = \argmin_{p\in \Uc_0} \, \Ac_0^{(p, \theta)} \, v_0$.  
All together proves that $\pb^\theta \in \Uc_0$ is admissible. Therefore, by Theorem \ref{thm:veri_in}, $\pb^\theta$ is indeed the insurer's optimal strategy, and $\Vb_0^\theta = v_0$ is the value function.
\end{proof}

\begin{proof}
	[Proof of Corollary \ref{prop:sen}]
	By recalling $\varphi_i$ in \eqref{eq:Theta}, $\varphi_i := \varphi_i(x| \delta_0, \delta_1, \delta_2, \lam_1, \lam_2)$ also depends on the parameters $\delta_0, \delta_1, \delta_2, \lam_1, \lam_2$. Although previously we have treated $\varphi_i$ as a univariate function of $x$ in Section \ref{sec:eq}, we will treat $\varphi_i$ as a multivariate function of these parameters in the subsequent sensitivity analysis. As such, $\varphi_i'(x)$ is Section \ref{sec:eq} translates to $\partial \varphi_i / \partial x$ here.

	First, we compute the partial derivatives of $\varphi_i$ with respect to each parameter by 
	\begin{align}
		\frac{\partial \varphi_i}{\partial \delta_0} &= \frac{2x^4}{(2x^2+((1+2\lambda_j)\delta_0+2\lambda_j\delta_i)x+\lambda_j(1+\lambda_j)\delta_0\delta_i)^2}>0 , \\
		\frac{\partial \varphi_i}{\partial \delta_i} &= \frac{x^2(\delta_0(1+\lambda_j)+2x)}{(2x^2+((1+2\lambda_j)\delta_0+2\lambda_j\delta_i)x+\lambda_j(1+\lambda_j)\delta_0\delta_i)^2}>0, \\
		\frac{\partial \varphi_i}{\partial \delta_j} &= 0, \\
		\frac{\partial \varphi_i}{\partial \lambda_i} &= 0, \\
		\frac{\partial \varphi_i}{\partial \lambda_j} &= -\frac{2(\delta_0^2+2\delta_0\delta_i+2\delta_i^2)x^3+2\delta_0\delta_i(\delta_0+2\delta_i)(1+\lambda_j)x^2+\delta_0^2\delta_i^2(1+\lambda_j)^2x}{(2x^2+((1+2\lambda_j)\delta_0+2\lambda_j\delta_i)x+\lambda_j(1+\lambda_j)\delta_0\delta_i)^2} < 0 .
	\end{align}

Because $\theta_i^* = \varphi_i(\theta_j^*)$, we have 
\begin{align}
	\frac{\partial\theta_i^*}{\partial\delta_0}= \frac{\partial \varphi_i }{\partial x}\Big|_{x=\theta_j^*} \cdot \frac{\partial\theta_j^*}{\partial\delta_0}+\frac{\partial \varphi_i}{\partial\delta_0}\Big|_{x=\theta_j^*}.
\end{align}
Solving the above equations, we get 
\begin{align}
	\frac{\partial\theta_i^*}{\partial\delta_0} = \frac{\frac{\partial \varphi_i }{\partial x}\Big|_{x=\theta_j^*} \cdot \frac{\partial \varphi_j }{\partial \delta_0}\Big|_{x=\theta_i^*} + \frac{\partial \varphi_i }{\partial \delta_0}\Big|_{x=\theta_j^*} }{1 - \frac{\partial \varphi_i }{\partial x}\Big|_{x=\theta_j^*} \cdot \frac{\partial \varphi_j }{\partial x}\Big|_{x=\theta_i^*}}> 0,
\end{align}
in which we have used \eqref{eq:derivative2}, \eqref{eq:fixed_point} and \eqref{eq:deri_cond} to derive the inequality
\begin{align}
	\frac{\partial \varphi_i }{\partial x}\Big|_{x=\theta_j^*} \cdot \frac{\partial \varphi_j }{\partial x}\Big|_{x=\theta_i^*} < 1.
\end{align}
Following a similar argument, we show the sensitivity results of $\theta_i^*$ in \eqref{eq:theta_sen} hold.

Next, we aim to analyze the sensitivity of $p^*$ with respect to those model parameters. To that end, recall that $p_i^* = \pb_i^{\theta^*}$, in which $\pb_1^\theta$ and $\pb_2^\theta$ are given by \eqref{eq:pb_1} and \eqref{eq:pb_2}, respectively. Similar to the above, we now treat each $p_i^*$ as a function of $\theta_1^*$, $\theta_2^*$, and model parameters, and compute the following partial derivatives: 
\begin{align}
	\frac{\partial p_i^*}{\partial \delta_0} &= \frac{2\theta_i^*(\theta_j^*)^2}{(2\theta_1^*\theta_2^*+\delta_0(\theta_1^*+\theta_2^*))^2}>0, 
	\\
	\frac{\partial p_i^*}{\partial \theta_i^*} &= -\frac{\delta_0\theta_j^*(\delta_0+2\theta_j^*)}{(2\theta_1^*\theta_2^*+\delta_0(\theta_1^*+\theta_2^*))^2}<0, 
	\\
	\frac{\partial p_i^*}{\partial \theta_j^*} &= \frac{\delta_0^2\theta_i^*}{(2\theta_1^*\theta_2^*+\delta_0(\theta_1^*+\theta_2^*))^2}>0.
\end{align}
Let us study the impact of $\delta_0$ on $p_1^*$ first and note 
\begin{align}
	\frac{\partial p_1^*}{\partial \delta_0} = \frac{\partial \pb_1}{\delta_0} \Big|_{\theta = \theta^*} + \frac{\partial \pb_1}{\partial \theta_1^*}\Big|_{\theta = \theta^*} \cdot \frac{\partial \theta_1^*}{\partial \delta_0} + \frac{\partial \pb_1}{\partial \theta_2^*}\Big|_{\theta = \theta^*} \cdot \frac{\partial \theta_2^*}{\partial \delta_0},
\end{align}
in which the first and third terms are positive, but the second term is negative. As such, a definite monotonicity result of $p_1^*$ with respect to $\delta_0$ is not available in general. By following the same argument, one can see that such a negative result applies to all partial derivatives of $p_i^*$.
\end{proof}

\begin{proof}[Proof of Corollary \ref{cor:limit}]
	{Since both $\lam_1$ and $\lam_2$ are positive constants, let us fix $\lam_1 > 0$ and consider $\lam_2 \nearrow \frac{1}{\lam_1}$ in the proof. By \eqref{eq:theta_sen}, for a fixed $\lam_1$, both $\theta_1^*$ and $\theta_2^*$ are strictly decreasing with respect to $\lam_2$. Therefore, when $\lam_2$ approaches $ \frac{1}{\lambda_1}$ form below, both $\theta_1^*$ and $\theta_2^*$ converge (decreasingly) to some nonnegative number, which we denote by $\theta^\circ_1$ and $\theta^\circ_2$, respectively. Our goal is then to show that $\theta^\circ_1 = \theta^\circ_2 = 0$. In what follows, we only show that $\theta^\circ_1 = 0$, as  the argument also applies to the proof of $\theta^\circ_2 = 0$.}
	
	{Assume to the contrary that $\theta^\circ_1 > 0$. Because the risk aversion parameters ($\delta_0, \delta_1, \delta_2$) and $\lam_1$ are fixed, and only $\lam_2$ varies, we will write functions $\varphi_i(\cdot)$ and $\varphi_i^{-1}(\cdot)$ as $\varphi_i(\cdot| \lam_2 )$ and $\varphi_i^{-1}(\cdot | \lam_2)$ to emphasize their dependence on $\lam_2$. From the proof of Theorem \ref{thm:eq}, we have 
		\begin{align*}
			\varphi_2(x|\lambda_2)>\varphi_1^{-1}(x|\lambda_2) \quad \text{for all } x \in (0, \theta_1^*).
		\end{align*}
		From the assumption $\theta^\circ_1 > 0$, we deduce $0<\theta^\circ_1/2<\theta^\circ_1\leq \theta^*_1$, which, together with the above result, implies 
		\begin{align}
			\label{eq:circ}
			\varphi_2(\theta^\circ_1/2|\lambda_2)>\varphi_1^{-1}(\theta^\circ_1/2|\lambda_2). 
		\end{align}
		However, using  the properties of $\varphi_1$ and $\varphi_2$ (recall their definitions in \eqref{eq:Theta}), we obtain, for $0 < \lam_1 \lam_2 < 1$, that 
		\begin{align*}
			\lim_{\lambda_2\nearrow\frac{1}{\lambda_1}} \varphi_1^{-1} (\theta^\circ_1/2|\lambda_2) &= \varphi_1^{-1} (\theta^\circ_1/2| 1 / \lam_1) > \varphi_2 (\theta^\circ_1/2| 1 / \lam_1) = \lim_{\lambda_2\nearrow\frac{1}{\lambda_1}} \varphi_2^{-1} (\theta^\circ_1/2|\lambda_2),
		\end{align*}
		a contradiction to \eqref{eq:circ}. Therefore, $\theta^\circ_1 = 0$ holds.}
	
	{When $\theta_1 \nearrow 0$ and $\theta_2 \nearrow 0$, the variance premium principle reduces to the actuarially fair premium, and a risk averse insurer will buy \emph{full} insurance, that is, $\lim_{\lam_1 \lam_2 \nearrow 1} \, (p_1^* + p_2^*) = 1$.}
\end{proof}

\bibliographystyle{apalike}
\bibliography{reference}

\begin{thebibliography}{}

\bibitem[Albrecher et~al., 2017]{albrecher2017reinsurance}
Albrecher, H., Beirlant, J., and Teugels, J.~L. (2017).
\newblock {\em Reinsurance: Actuarial and Statistical Aspects}.
\newblock John Wiley \& Sons.

\bibitem[Arrow, 1963]{arrow1963uncertainty}
Arrow, K.~J. (1963).
\newblock Uncertainty and the welfare economics of medical care.
\newblock {\em American Economic Review}, 53(5):941--973.

\bibitem[Asmussen et~al., 2000]{asmussen2000optimal}
Asmussen, S., H{\o}jgaard, B., and Taksar, M. (2000).
\newblock Optimal risk control and dividend distribution policies. {Example} of
  excess-of loss reinsurance for an insurance corporation.
\newblock {\em Finance and Stochastics}, 4(3):299--324.

\bibitem[Bensoussan et~al., 2014]{bensoussan2014class}
Bensoussan, A., Siu, C.~C., Yam, S. C.~P., and Yang, H. (2014).
\newblock A class of non-zero-sum stochastic differential investment and
  reinsurance games.
\newblock {\em Automatica}, 50(8):2025--2037.

\bibitem[Bo et~al., 2024]{LSC2024}
Bo, L., Wang, S., and Zhou, C. (2024).
\newblock A mean field game approach to optimal investment and risk control for
  competitive insurers.
\newblock {\em Insurance: Mathematics and Economics}, 116:202--217.

\bibitem[Boonen et~al., 2021]{reinsuranceTJ}
Boonen, T.~J., Tan, K.~S., and Zhuang, S.~C. (2021).
\newblock Optimal reinsurance with multiple reinsurers: {Competitive} pricing
  and coalition stability.
\newblock {\em Insurance: Mathematics and Economics}, 101(B):302--319.

\bibitem[Cai and Chi, 2020]{cai2020optimal}
Cai, J. and Chi, Y. (2020).
\newblock Optimal reinsurance designs based on risk measures: A review.
\newblock {\em Statistical Theory and Related Fields}, 4(1):1--13.

\bibitem[Cao et~al., 2022]{StackelbergBinZou}
Cao, J., Li, D., Young, V.~R., and Zou, B. (2022).
\newblock Stackelberg differential game for insurance under model ambiguity.
\newblock {\em Insurance: Mathematics and Economics}, 106:128--145.

\bibitem[Cao et~al., 2023a]{cao2023breinsurance}
Cao, J., Li, D., Young, V.~R., and Zou, B. (2023a).
\newblock Reinsurance games with n variance-premium reinsurers: {From} tree to
  chain.
\newblock {\em ASTIN Bulletin}, 53(3):706--728.

\bibitem[Cao et~al., 2023b]{cao2023areinsurance}
Cao, J., Li, D., Young, V.~R., and Zou, B. (2023b).
\newblock Reinsurance games with two reinsurers: Tree versus chain.
\newblock {\em European Journal of Operational Research}, 310(2):928--941.

\bibitem[Chen et~al., 2016]{CQSW2016}
Chen, L., Qian, L., Shen, Y., and Wang, W. (2016).
\newblock Constrained investment–reinsurance optimization with regime
  switching under variance premium principle.
\newblock {\em Insurance: Mathematics and Economics}, 71:253--267.

\bibitem[Chen and Shen, 2018]{StackelbergLvYang}
Chen, L. and Shen, Y. (2018).
\newblock On a new paradigm of optimal reinsurance: A stochastic stackelberg
  differential game between an insurer and a reinsurer.
\newblock {\em ASTIN Bulletin}, 48(2):905--960.

\bibitem[Chen and Shen, 2019]{chen2019stochastic}
Chen, L. and Shen, Y. (2019).
\newblock Stochastic {Stackelberg} differential reinsurance games under
  time-inconsistent mean--variance framework.
\newblock {\em Insurance: Mathematics and Economics}, 88:120--137.

\bibitem[Chi, 2012]{chi2012optimal}
Chi, Y. (2012).
\newblock Optimal reinsurance under variance related premium principles.
\newblock {\em Insurance: Mathematics and Economics}, 51(2):310--321.

\bibitem[Fleming and Soner, 2006]{fleming2006controlled}
Fleming, W.~H. and Soner, H.~M. (2006).
\newblock {\em Controlled Markov Processes and Viscosity Solutions}.
\newblock Springer Science \& Business Media.

\bibitem[Gu et~al., 2018]{GVY2018}
Gu, A., Viens, F.~G., and Yao, H. (2018).
\newblock Optimal robust reinsurance-investment strategies for insurers with
  mean reversion and mispricing.
\newblock {\em Insurance: Mathematics and Economics}, 80:93--109.

\bibitem[Jiang et~al., 2019]{JRYH2019}
Jiang, W., Ren, J., Yang, C., and Hong, H. (2019).
\newblock On optimal reinsurance treaties in cooperative game under
  heterogeneous beliefs.
\newblock {\em Insurance: Mathematics and Economics}, 85:173--184.

\bibitem[Jin et~al., 2024]{jin2024optimal}
Jin, Z., Xu, Z.~Q., and Zou, B. (2024).
\newblock Optimal moral-hazard-free reinsurance under extended distortion
  premium principles.
\newblock {\em SIAM Journal on Control and Optimization}, 62(3):1390--1416.

\bibitem[Kroell et~al., 2023]{kroell2023optimal}
Kroell, E., Jaimungal, S., and Pesenti, S.~M. (2023).
\newblock Optimal robust reinsurance with multiple insurers.
\newblock {\em arXiv preprint arXiv:2308.11828}.

\bibitem[Li et~al., 2015]{LRZ2015}
Li, D., Rong, X., and Zhao, H. (2015).
\newblock Time-consistent reinsurance–investment strategy for a
  mean–variance insurer under stochastic interest rate model and inflation
  risk.
\newblock {\em Insurance: Mathematics and Economics}, 64:28--44.

\bibitem[Li and Young, 2022]{DV2022}
Li, D. and Young, V.~R. (2022).
\newblock Stackelberg differential game for reinsurance: {Mean-variance}
  framework and random horizon.
\newblock {\em Insurance: Mathematics and Economics}, 102:42--45.

\bibitem[Liang and Yuen, 2016]{liang2016optimal}
Liang, Z. and Yuen, K.~C. (2016).
\newblock Optimal dynamic reinsurance with dependent risks: {Variance} premium
  principle.
\newblock {\em Scandinavian Actuarial Journal}, 2016(1):18--36.

\bibitem[Lu et~al., 2016]{LMWS2016}
Lu, Z., Meng, L., Wang, Y., and Shen, Q. (2016).
\newblock Optimal reinsurance under {VaR} and {TVaR} risk measures in the
  presence of reinsurer’s risk limit.
\newblock {\em Insurance: Mathematics and Economics}, 68:92--100.

\bibitem[Meng et~al., 2017]{meng2017note}
Meng, H., Siu, T.~K., and Yang, H. (2017).
\newblock A note on optimal insurance risk control with multiple reinsurers.
\newblock {\em Journal of Computational and Applied Mathematics}, 319:38--42.

\bibitem[Peng et~al., 2021]{PCW2021}
Peng, X., Chen, F., and Wang, W. (2021).
\newblock Robust optimal investment and reinsurance for an insurer with inside
  information.
\newblock {\em Insurance: Mathematics and Economics}, 96:15--30.

\bibitem[Schmidli, 2001]{schmidli2001optimal}
Schmidli, H. (2001).
\newblock Optimal proportional reinsurance policies in a dynamic setting.
\newblock {\em Scandinavian Actuarial Journal}, 2001(1):55--68.

\bibitem[Schmidli, 2002]{schmidli2002minimizing}
Schmidli, H. (2002).
\newblock On minimizing the ruin probability by investment and reinsurance.
\newblock {\em Annals of Applied Probability}, 12(3):890--907.

\bibitem[Tan et~al., 2020]{tan2020optimal}
Tan, K.~S., Wei, P., Wei, W., and Zhuang, S.~C. (2020).
\newblock Optimal dynamic reinsurance policies under a generalized
  {Denneberg’s} absolute deviation principle.
\newblock {\em European Journal of Operational Research}, 282(1):345--362.

\bibitem[Tan et~al., 2009]{tan2009var}
Tan, K.~S., Weng, C., and Zhang, Y. (2009).
\newblock {VaR} and {CTE} criteria for optimal quota-share and stop-loss
  reinsurance.
\newblock {\em North American Actuarial Journal}, 13(4):459--482.

\bibitem[Wang et~al., 2020]{GYS2020}
Wang, G., Wang, Y., and Zhang, S. (2020).
\newblock An asymmetric information mean-field type linear-quadratic stochastic
  stackelberg differential game with one leader and two followers.
\newblock {\em Optimal Control Applications and Methods}, 41(4):1001--1370.

\bibitem[Yao and Zhu, 2024]{yao2024optimal}
Yao, D. and Zhu, J. (2024).
\newblock Optimal reinsurance under a new design: two layers and multiple
  reinsurers.
\newblock {\em Quantitative Finance}, 24(5):655--676.

\bibitem[Yong and Zhou, 2012]{yong2012stochastic}
Yong, J. and Zhou, X.~Y. (2012).
\newblock {\em Stochastic Controls: Hamiltonian Systems and HJB Equations}.
\newblock Springer Science \& Business Media.

\bibitem[Zhang and Li, 2021]{ZL2021}
Zhang, L. and Li, B. (2021).
\newblock Optimal reinsurance under the $\alpha$-maxmin mean-variance
  criterion.
\newblock {\em Insurance: Mathematics and Economics}, 101:225--239.

\bibitem[Zhou and Yuen, 2012]{ZY2012}
Zhou, M. and Yuen, K.~C. (2012).
\newblock Optimal reinsurance and dividend for a diffusion model with capital
  injection: Variance premium principle.
\newblock {\em Economic Modelling}, 29(2):198--207.

\bibitem[Zhu et~al., 2023]{MMT2023}
Zhu, M.~B., Ghossoub, M., and Boonen, T.~J. (2023).
\newblock Equilibria and efficiency in a reinsurance market.
\newblock {\em Insurance: Mathematics and Economics}, 113:24--49.

\end{thebibliography}

\end{document}